\documentclass[11pt,reqno,oneside]{amsart}

\usepackage[margin=2cm]{geometry}

\usepackage{amsmath,amssymb,amscd,amsthm}
\usepackage{mathrsfs}
\usepackage{mathtools}
\usepackage{enumerate}
\usepackage{cite}
\usepackage{comment}
\usepackage[colorlinks,pdfstartview=FitB]{hyperref}
\hypersetup{allcolors=blue}
\usepackage{graphicx}
\usepackage{xcolor}
\usepackage{float}

\usepackage{tikz}
\usetikzlibrary{calc,
	matrix,
	fadings,
	fit,
	math,
	arrows.meta}
\usepackage{pgfplots}
\pgfplotsset{compat=1.18}

\usepackage{bbm}

\def\idty{\mathbbm{1}} 

\newcommand{\Z}{{\mathbb Z}}

\newcommand{\C}{{\mathbb C}}

\newcommand{\N}{{\mathbb N}}

\newcommand{\CH}{{\mathcal{H}}}

\newcommand{\CL}{{\mathcal{L}}}
\newcommand{\CM}{{\mathcal{M}}}

\newcommand{\CU}{{\mathcal{U}}}

\newcommand{\Span}{{\mathrm{span}}}

\newtheorem{theorem}{Theorem}[section]
\newtheorem{lemma}[theorem]{Lemma}

\newtheorem{coro}[theorem]{Corollary}

\newtheorem{remark}[theorem]{Remark}

\theoremstyle{definition}

\allowdisplaybreaks \numberwithin{equation}{section}

\DeclareMathOperator{\supp}{supp}

\newcommand{\restr}[2]{#1\big|_{#2}}

\makeatletter
\def\subsection{\@startsection{subsection}{2}%
	\z@{.5\linespacing\@plus.7\linespacing}{.5\linespacing}%
	{\normalfont\scshape\centering}}
\makeatother

\hypersetup{
	pdftitle = {Absence of ballistic transport in quantum walks},
	pdfauthor = {Houssam Abdul-Rahman, Thomas A. Jackson,
		Yousef Salah},
	pdfsubject = { },
	pdfkeywords = {Ballistic transport, Localization,  Quantum walk, CMV matrix}
}

\makeatletter
\renewcommand{\section}{\@startsection
  {section}{1}{\z@}%
  {-2.5ex \@plus -1ex \@minus -.2ex}%
  {1ex \@plus .2ex}%
  {\normalfont\Large\bfseries}}

\renewcommand{\subsection}{\@startsection
  {subsection}{2}{\z@}%
  {-2ex \@plus -1ex \@minus -.2ex}%
  {0.8ex \@plus .2ex}%
  {\normalfont\large\bfseries}}
\makeatother

\begin{document}

\title[Zero-velocity in Asymptotically Reflecting Quantum Walks]{Absence of Ballistic Transport in Quantum Walks with Asymptotically Reflecting Sites}
		
\author[H.\ Abdul-Rahman]{Houssam Abdul-Rahman}
\address{[H.\ Abdul-Rahman] Department of Mathematical Sciences, United Arab Emirates University, AL Ain, UAE}
\email{\href{mailto:houssam.a@uaeu.ac.ae}{houssam.a@uaeu.ac.ae}}


\author[T. Jackson]{Thomas A. Jackson}
\address{[T. Jackson] Department of Mathematical Sciences, United Arab Emirates University, AL Ain, UAE}
\email{\href{mailto:andjackson@uaeu.ac.ae}{andjackson@uaeu.ac.ae}}

\author[Y. Salah]{Yousef Salah}
\address{[Y. Salah] Department of Mathematical Sciences, United Arab Emirates University, AL Ain, UAE}
\email{\href{mailto:202050242@uaeu.ac.ae}{202050242@uaeu.ac.ae}}

\begin{abstract}
We prove general sufficient conditions for zero velocity in position dependent one-dimensional quantum walks, and hence for the absence of ballistic transport. Our starting point is a general a priori upper bound on the velocity, formulated in terms of sparse bi-infinite sequences of sites, their gap structure, and the corresponding local coin parameters. This estimate yields several deterministic criteria for zero velocity that depend only on the behavior of suitable coin entries along selected subsequences and are independent of the values of the coins elsewhere. We also discuss the random case as an application of the general approach. All of our results remain valid in the CMV setting.
\end{abstract}
\maketitle



%

\allowdisplaybreaks

\tableofcontents
\section{Introduction}
Discrete-time quantum walks constitute a natural unitary framework for studying transport on discrete structures, and they are of interest both because of their mathematical richness and because of their relevance in quantum science. Beyond their role as quantum analogues of random walks \cite{QD-1}, they arise in quantum algorithms, computation, and simulation \cite{Portugal2013,App-QGates1,App-QS1}. At the same time, recent controlled realizations, particularly in photonic platforms, have made them an experimentally accessible setting for studying transport and related phenomena \cite{Karski2009,Esposito2022,Zhou2024}. They have also found important applications in topological phases of driven systems \cite{Kitagawa2010}. In one dimension, they furnish a particularly tractable framework for spectral and dynamical questions concerning propagation \cite{BHJ2003,ahlbrechtAsymptoticEvolutionQuantum2011}.

Depending on the structure of the system, one may observe markedly different
transport regimes. Periodic settings are typically associated with ballistic transport
\cite{ARS2023-CMP,ARCSW25,CFLOZ, damanikQuantumDynamicsPeriodic2015, Nguyen2020, ewalks},
static disordered settings often exhibit localization or strong suppression of
transport \cite{BHJ2003,HamzaJoyeStolz2006,hamzaLyapunovExponentsUnitary2007,HJS2009,HJ14,AschJoye2019,locQuasiPer},
and quasi-periodic settings may display intermediate behavior, including anomalous
or diffusive-type transport \cite{Nguyen2019,Wojcik}.
 At the spectral level, ballistic transport is generally linked to absolutely continuous spectrum, whereas localization is tied to pure point spectrum; see, e.g., \cite{KumarSabri2026,CedzichFillmanVelazquez2026}. Moreover, one-dimensional coined quantum walks admit a natural representation in terms of CMV matrices and in the broader class of generalized extended CMV matrices, thereby linking quantum-walk dynamics with the spectral theory of orthogonal polynomials on the unit circle \cite{canteroFivediagonalMatricesZeros2003,Simon2005OPUC2,canteroMatrixvaluedSzegoPolynomials2010,CGMV2012QIP,CFLOZ, joye2026dynamical, Simon2007CMV}.

In the present work, we study general one-dimensional inhomogeneous split-step quantum
walks and seek mechanisms ensuring the absence of ballistic transport. Our main
dynamical quantity is the velocity, which measures the long-time propagation rate of
wave packets \cite{ARFFW,ARS2023-CMP,ADFS2024,ARCSW25}; positive velocity corresponds
to ballistic propagation, whereas zero velocity rules out linear spreading.

The basic motivating picture already appears in the simpler shift-coin setting
$U=SC$, where $S$ and $C$ denote the shift and coin operators. A perfect reflector
corresponds to a coin whose $(1,1)$-entry vanishes, so that the incoming amplitude is
fully reflected rather than transmitted. By contrast, a perfect transmitter
corresponds to a trivial coin action, for instance $C=\idty_2$, in which case no
mixing occurs and the wave is fully transmitted. If two perfect reflectors are placed on the line, then any state initially supported
between them remains trapped in the cavity they form and cannot propagate outside it.
More generally, a bi-infinite sequence of perfect reflectors decomposes the line into
finite cavities, and therefore prevents ballistic transport. This naturally leads to
the question of whether the same mechanism persists when the perfect reflectors are
replaced by imperfect reflectors converging, in both directions, to a perfect reflector.

The main results of the paper identify a deterministic mechanism for the
absence of ballistic transport in this setting. Theorem~\ref{Thm:general} gives a
general upper bound on the velocity in terms of the decay of the relevant coin
entries along a bi-infinite sequence $(j_m)_{m\in\Z}\subset\Z$ and the geometry of the gaps
between consecutive sites in that sequence. Theorem~\ref{Thm:determ} turns this
general bound into explicit zero-velocity criteria in several natural regimes.
In particular, sufficiently sparse families of sufficiently strong reflectors force
vanishing velocity.

A key feature of the result is that it depends only on the behavior of the coins along the
distinguished sequence $(j_m)_{m\in\Z}$, and is otherwise insensitive to the values of the
coins away from that set. In this sense, the criterion is local and sparse in nature. The zero velocity conclusion rules out ballistic
transport, although it leaves open the question about slower transport or localization. 
As an application of the deterministic framework, we also obtain a simple almost-sure
zero-velocity criterion in an i.i.d.\ random setting.

To the best of our knowledge, there is no previous work establishing deterministic
zero-velocity criteria, or equivalently the absence of ballistic transport, for
bi-infinite quantum walks of the present local sparse-reflector type.

The main novelty is not only in the statements, but also in the methods behind them. We develop a new method for bounding transport, based on the introduction of a modified position operator tailored to a sparse decomposition of the walk into cavities. This allows one to replace the standard position observable by a block-scalar operator that commutes with the inactive part of the dynamics and reduces the problem to estimating only the interfacial couplings across the chosen sparse almost reflector sites. The resulting a priori bound is therefore subsequence-driven and insensitive to the coin values away from that subsequence. To the best of our knowledge, this modified-position-operator framework is new in the present quantum-walk/CMV setting and provides the structural reason that the absence of ballistic transport can be proved from such limited local information.

It is also natural to ask whether our approach can be applied to quasi-periodic
models, such as unitary analogues of the almost-Mathieu operator or Fibonacci quantum
walks; see, e.g.,
\cite{CFO1,FOZ2016,CFGW2020LMP,WangDamanik2019,ZhangPiao2025,damanikSpreadingEstimatesQuantum2016,Nguyen2020,Nguyen2019}.
Indeed, this was one of the motivations for the present work. However, we were unable to determine how to adapt the methods developed here to prove or disprove vanishing velocity in these settings. This may indicate that, in quasi-periodic models, the relevant mechanism is more global and interference-driven, and is therefore not immediately accessible to our present approach.

Although our focus here is on the absence of ballistic transport rather than on the spectral classification of the model, the present setting is closely related to the literature on sparse barrier models, where widely separated barriers have a decisive influence on both transport and spectral type. Much of this literature is also motivated by the appearance, and in some cases the construction, of singular continuous spectrum. On the unitary side, the closest connection is with the half-line CMV model with sparse growing barriers studied by Damanik et al.~\cite{DamanikEtAl2016}. Analogous sparse-barrier mechanisms have also been studied for self-adjoint Schr\"odinger operators, beginning with Pearson \cite{Pearson1978} and continuing in works such as Simon--Stolz \cite{SimonStolz1996}, Zlato\v{s} \cite{Zlatos2004}, and Tcheremchantsev \cite{Tcheremchantsev2004}. From a more physical perspective, our setting is also related to the defect-scattering literature for one-dimensional quantum walks, where localized inhomogeneities are known to induce reflection, trapping, and suppression of transport \cite{LiIzaacWang2013,CanteroGrunbaumMoralVelazquez2010,WojcikEtAl2012,ZhangXueTwamley2014}.

The paper is organized as follows. The next section introduces the general setting, including the motivating example of a sequence of perfect reflectors, and states the main deterministic results, Theorems~\ref{Thm:determ} and \ref{Thm:general}, together with the application to the disordered case given in Corollary~\ref{coro:random}. The proofs occupy the remainder of the paper. In particular, Section~\ref{sec:properties} introduces the key idea of a modified position operator. Section~\ref{sec:pf-general} contains the proof of Theorem~\ref{Thm:general}, while Sections~\ref{sec:zero-velocity-1-2} and \ref{sec:pf:case3} are devoted to the proof of Theorem~\ref{Thm:determ}. Finally, Section~\ref{sec:pf-random} treats the random case.

\section*{Acknowledgments}
The authors would like to thank Christopher Cedzich, Jake Fillman, Darren Ong, and Mostafa Sabri for helpful and stimulating discussions.\\
H. A.-R. was supported in part by the UAE University under grant number 12S161-G00004622.

\section{Setting and main results}

We consider the dynamics of discrete-time quantum walks on the one-dimensional
lattice. The walk acts on the Hilbert space
\begin{equation}
  \mathcal{H} = \ell^2(\Z)\otimes \C^2 .
\end{equation}
Let $\{|j\rangle : j\in\Z\}$ and $|+\rangle = \begin{bmatrix}1\\[0.2em]0\end{bmatrix},\ 
  |-\rangle = \begin{bmatrix}0\\[0.2em]1\end{bmatrix}$ be the canonical orthonormal basis of
$\ell^2(\Z)$ and $\C^2$, respectively. The corresponding standard basis of
$\mathcal{H}$ is then given by
$\delta_j^\pm := |j\rangle \otimes |\pm\rangle,\ j\in\Z$.

We study the general \emph{split-step} walks on $\CH$, of the form
\begin{equation}\label{def:W}
W=S_+ C_1 S_- C_2
\end{equation}

where $C_1$ and $C_2$ are coin operators acting locally on the internal (coin) degrees of
freedom and $S_\pm$ are conditional shifts operators. More precisely:

\begin{enumerate}[(i)]
\item The coin operators $C_1$ and $C_2$ are specified by a bi-infinite sequences $a_1, a_2:\Z\to\C$
\begin{equation}\label{def:a}
  a_k = (a_k(n))_{n\in\Z}
  = (\ldots, a_k(-2),a_k(-1),a_k(0),a_k(1),a_k(2),\ldots), \text{ for }k=1,2,
\end{equation}
and corresponding sequences $b_1, b_2:\Z\to\C$, such that at position
$n\in\Z$ the coins act as the $2\times 2$ unitary matrix
\begin{equation}\label{def:C}
  C_k(n) =
  \begin{bmatrix}
    a_k(n) & b_k(n)\\[0.2em]
   -\overline{b_k(n)} & \overline{ a_k(n)}
  \end{bmatrix},
  \qquad
  |a_k(n)|^2 + |b_k(n)|^2 = 1.
\end{equation}
Thus $C_k$ is the direct sum of the local coins $C_k(n)$ in the position
representation.

\item The shift operators $S_\pm$ on $\mathcal{H}$ are defined by
\begin{equation}\label{def:S}
  S_+ = T \otimes |+\rangle\langle +|
      + \idty_\Z \otimes |-\rangle\langle -|\ \text{and}\ S_- = \idty_\Z\otimes |+\rangle\langle +|
      + T^{-1} \otimes |-\rangle\langle-|
\end{equation}
where $T^{\pm1}$ denotes the translation-by-one operator on $\ell^2(\Z)$,
$T^{\pm 1}|j\rangle = |j\pm 1\rangle,  j\in\Z$.
\end{enumerate}
When $C_1=\idty_\CH$ in \eqref{def:W}, then the quantum walk reduces to a \emph{shift-coin} walk $U=SC$ where $S:=S_+ S_-$.

By identifying $\ell^2(\Z)\otimes\C^2$ with $\ell^2(\Z)$ through the mapping
\begin{equation}\label{def:identify}
\delta_j^+\mapsto \delta_{2j-1} \text{ and }\delta_j^-\mapsto \delta_{2j},
\end{equation}
the split-step walk $W=S_+ C_1 S_- C_2$ can be regarded as a \emph{generalized extended CMV} matrix 
\begin{equation}\label{def:CMV}
W\equiv \mathcal{L}\mathcal{M},\footnote{``$\equiv$'' means that there exists a unitary $\CU$ such that $\CU W\CU^*=\CL\CM$. In this case, $\CU:\ell^2(\Z)\otimes \C^2\to\ell^2(\Z)$ defined by \eqref{def:identify}.} \text{ where } \CL:=\CL(a_1,b_1)=\bigoplus_{n\in\Z}\Theta_1(n), \qquad \CM:=\CM(a_2,b_2)=\bigoplus_{n\in\Z}\Theta_2(n).
\end{equation}
Here $\Theta_k(n):=\sigma_x C_k(n)$ denotes the $n$-th $2\times2$ block in the corresponding direct sum, and
$\sigma_x=\begin{bmatrix}0 & 1\\1 & 0\end{bmatrix}$
is the Pauli $x$-matrix.
More precisely, $\mathcal{L}$ is block-diagonal with respect to the decomposition
$
\ell^2(\Z)=\bigoplus_{n\in\Z}\Span\{\delta_{2n},\delta_{2n+1}\}$,
whereas $\mathcal{M}$ is block-diagonal with respect to
$\ell^2(\Z)=\bigoplus_{n\in\Z}\Span\{\delta_{2n-1},\delta_{2n}\}$.
The latter is exactly the decomposition associated with the lattice sites in the quantum-walk picture; in particular,
$\mathcal{M}\equiv \sigma_x^{\oplus\Z} C_2$, see, e.g., \cite{ARCSW25}.

\begin{remark}
Consider the standard notation of  the generalized CMV blocks
\begin{equation}
\Theta(\alpha,\rho):=
\begin{bmatrix}
\overline{\alpha} & \overline{\rho}\\
\rho & -\alpha
\end{bmatrix},
\qquad |\alpha|^2+|\rho|^2=1.
\end{equation}
Since
\begin{equation}
\Theta_k(n)=\sigma_x C_k(n)=
\begin{bmatrix}
-\overline{b_k(n)} & \overline{a_k(n)}\\
a_k(n) & b_k(n)
\end{bmatrix}
=\Theta\bigl(-b_k(n),a_k(n)\bigr),
\end{equation}
each coin $C_k(n)$ gives rise to the Verblunsky pair
$\bigl(-b_k(n),a_k(n)\bigr)$.
With the identification \eqref{def:identify}, the blocks of $\mathcal L$
act on $\Span\{\delta_{2n},\delta_{2n+1}\}$, corresponding to the even
CMV indices, whereas the blocks of $\mathcal M$ act on
$\Span\{\delta_{2n-1},\delta_{2n}\}$, corresponding to the odd ones.
Therefore,
\begin{equation}\label{def:Verblunsky}
(\alpha_{2n},\rho_{2n})=\bigl(-b_1(n),a_1(n)\bigr),
\qquad
(\alpha_{2n-1},\rho_{2n-1})=\bigl(-b_2(n),a_2(n)\bigr),
\qquad n\in\Z.
\end{equation}
In other words, the even Verblunsky pairs come from the $C_1$-coins,
while the odd Verblunsky pairs come from the $C_2$-coins.
\end{remark}

\subsection{Motivation: a simple case}\label{sec:motivation}
To motivate the main results, consider the shift-coin walk $U=SC$ on $\mathcal{H}$, where $C$ is specified by the bi-infinite sequence $a:\Z\to\C$.
If for some $\ell\in\Z$ we have $a(\ell) = 0$, that is, up to phases on the off diagonal entries,
\begin{equation}
  C(\ell) =
  \begin{bmatrix}
    0 & 1\\[0.2em]
   -1 & 0
  \end{bmatrix},
\end{equation}
then the site $\ell$ acts as a \emph{perfect reflector} in the sense
that a state cannot propagate from position $\ell-1$ to $\ell+1$, or vice
versa. This is a consequence of the following simple observation.

First, the shift operator $S$ expels any amplitude from position $\ell$:
\begin{equation}\label{U:mirror-1}
  SC\ \delta_\ell^+ = -\delta_{\ell-1}^-,
  \qquad
  SC\ \delta_\ell^- = \delta_{\ell+1}^+.
\end{equation}
Moreover, a state that arrives at position $\ell$ from either side is sent
back after two steps of the walk (see also Figure \ref{fig:mirror}):\footnote{Observe that
$
(SC)^2\delta_{\ell-1}^+ = SCS\left(\alpha\delta_{\ell-1}^+ +\beta\delta_{\ell-1}^-\right), \text{ where }|\alpha|^2+|\beta|^2
=1$.
Under the action of $SCS$, the component $\delta_{\ell-1}^-$ is propagated to the left. Therefore, we restrict attention to the action of $SCS$ on $\delta_{\ell-1}^+ $ and $\delta_{\ell+1}^-$, as these are the components that are mapped to the reflector at site $\ell$.}
\begin{equation}\label{U:mirror-2}
  S C S\,\delta_{\ell-1}^+ = -\delta_{\ell-1}^-,
  \qquad
  S C S\,\delta_{\ell+1}^- = -\delta_{\ell+1}^+.
\end{equation}

In both cases the state is reflected, up to a global phase.
\vspace{-0.4cm}
\begin{figure}[H]\label{fig:mirror}
\begin{center}
\includegraphics[width=1.5in]{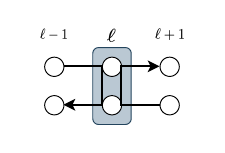}
\vspace{-0.7cm}
\caption{Each vertical pair of circles represents a site of the walk together with its two internal degrees of freedom. The upper circle corresponds to $|+\rangle$, while the lower circle corresponds to $|-\rangle$. The figure shows the sites $\ell-1$, $\ell$ (boxed), and $\ell+1$.}
\end{center}
\end{figure}

Consequently, if two such reflectors are present, then any state initially
supported between them remains trapped in the corresponding finite region
as time evolves. On the other hand, if the initial state is supported
outside this \emph{cavity}, then the state propagates away from the nearest
reflector with positive velocity. We define the \emph{maximal velocity} \footnote{Strictly speaking, $v(W)$ is a \emph{maximal upper asymptotic velocity}: maximal because of the supremum over normalized states, upper because of the $\limsup$, asymptotic because it describes the regime $t\to\infty$, and velocity because it measures a large-time growth rate per unit time. Throughout the paper, however, we will use the shorter term \emph{maximal velocity}, or simply \emph{velocity}.} of a quantum walk generated by a unitary
operator $W$ on $\mathcal H$ as, see, e.g. \cite{ARCSW25}
\begin{equation}\label{def:v}
  v(W)
  = \sup_{\substack{\psi\in D(Q)\\ \|\psi\|=1}}
    \limsup_{t\to\infty}\frac{1}{t}\bigl\|Q\,W^t \psi\bigr\|,
\end{equation}
where $Q$ is the position operator on $\mathcal H$,
\begin{equation}\label{def:Q}
  Q = \sum_{j\in\Z} j\, |j\rangle\langle j|\otimes \idty_2,
\end{equation}
and $D(Q)$ is the domain of $Q$ consisting of all $\psi\in\mathcal{H}$ such that $\|Q\psi\|_\mathcal{H}<\infty$.

 Since the maximal velocity \eqref{def:v} is
defined as a supremum over all initial states in $D(Q)$,
the maximal velocity of $U=SC$ has to vanish, i.e., $v(U)=0$ when there is an (infinite) sequence of
reflectors in both spatial directions.

This setting inspires a natural decomposition of the Hilbert space, that will be crucial for most of our proofs, as follows. 
Suppose that the bi-infinite sequence of perfect reflectors are at sites $(j_m)_{m\in\Z}\subset\Z$ such that
\begin{equation}\label{def:j}
\ldots<j_{-1}<j_0<j_1<\ldots.
\end{equation}
We identify $\mathcal{H} = \ell^2(\Z)\otimes \C^2$
with $\ell^2(\Z)$ via \eqref{def:identify}. Under which, each pair of consecutive perfect reflectors at
$j_m$ and $j_{m+1}$ defines a finite \emph{cavity} in $\ell^2(\Z)$. For
$m\in\Z$ we set
\begin{equation}
  \mathcal{H}_m
  = \Span\{\delta_n : n = 2j_m,2j_m+1,\dots,2j_{m+1}-1\}.
\end{equation}
Then $ \ell^2(\Z)$ decomposes naturally as (see Figure \ref{fig:dec})
\begin{equation}\label{dec}
  \ell^2(\Z) = \bigoplus_{m\in\Z} \mathcal{H}_m.
\end{equation}
This is an orthogonal decomposition into finite-dimensional subspaces
corresponding to the cavities between consecutive perfect reflectors.
\begin{figure}[H]\label{fig:dec}
\begin{center}
\includegraphics[width=4in]{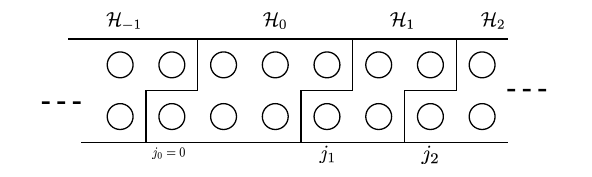}
\vspace{-0.5cm}
\caption{Schematic decomposition of $\ell^2(\mathbb{Z})$ into the subspaces $\mathcal{H}_m$ determined by the increasing sequence $(j_m)_{m\in\mathbb{Z}}$. Each vertical pair of circles represents the two basis vectors associated with a lattice site: the upper circle corresponds to $\delta_{2j_k-1}$, and the lower circle to $\delta_{2j_k}$. The indices $j_0=0$, $j_1$, and $j_2$ mark successive interface sites separating neighboring subspaces, with $j_0=0$ chosen as the origin.}
\end{center}
\end{figure}

By relabeling the indices, we may assume without loss of
generality that $j_0 = 0$, so that no block $\mathcal{H}_k$ contains
both positive and negative positions. 

In the CMV setting, $U=SC\equiv \CL\CM$, where $\CL\equiv S\sigma_x^{\oplus\Z}$, and $\CM\equiv\sigma_x^{\oplus\Z}C$ \footnote{In both $\CL$ and $\CM$, the direct sum $\sigma_x^{\oplus\Z}$ is compatible with that of $C$ (or $\CM$)}. Both of these operators are block diagonal with respect to \eqref{dec}. That is $U$ decouples into a direct sum of finite blocks. Hence, the system trivially exhibit all the usual manifestations of localization.

Here we give a quantum walk argument that shows that $v(U)=0$.
The set of compactly supported normalized states is dense in $D(Q)$, hence
\begin{equation}
  v(U)
  = \sup_{\substack{\supp(\psi_c)\ \text{compact} \\ \|\psi_c\|=1}}
    \limsup_{t\to\infty}\frac{1}{t}\bigl\|Q U^t \psi_c\bigr\|.
\end{equation}
Fix a compactly supported initial state $\psi_c$. Then there exists
$s\in\N$ such that
\begin{equation}
  \supp(\psi_c)\subset
  \mathbb{H}_s := \bigoplus_{-s\leq k\leq s}\mathcal{H}_k.
\end{equation}
Because $\CH_k$ are invariant under  $U^t$ (see \eqref{U:mirror-1} and \eqref{U:mirror-2}), then we also have
$\supp(U^t\psi_c)\subset \mathbb{H}_s$ for all $t\in\N$.
It follows that
\begin{equation}
  \frac{1}{t}\bigl\|Q U^t \psi_c\bigr\|
  \leq \frac{1}{t}\bigl\|\restr{Q}{\mathbb{H}_s}\bigr\| \leq \frac{1}{t}\max\{j_{s+1}, |j_{-s}|\}\to 0 \text{ as } t\to\infty.
\end{equation}
Because this holds for every compactly supported normalized $\psi_c$,
we conclude that $v(U)=0$.

In fact, the idea of bi-infinite sequence of perfect reflectors is valid in the more general split-step setting, and the corresponding generalized CMV matrices. 
\begin{lemma}\label{lem:case-a=0}
Consider the quantum walk $W = S_+ C_1 S_- C_2$ defined by the coins \eqref{def:C} and shifts
\eqref{def:S}, where the coins $C_k$ are associated with the sequences
$(a_k(n))_{n\in\Z}$ for $k=1,2$ as in \eqref{def:a}. For a fixed $k\in\{1,2\}$, suppose there exists a bi-infinite subsequence
$(a_k(j_m))_{m\in\Z}$ of zeros with strictly increasing indices as in \eqref{def:j}.
Then the maximal velocity $v(W)$ defined in \eqref{def:v} is zero, i.e.\ $v(W)=0$.
\end{lemma}
While Lemma \ref{lem:case-a=0} could be established by a straightforward extension of the preceding argument together with Lemma \ref{lem:com-velocity} below (or by the sieving Theorem \cite[Theorem 3.3]{ARCSW25}), we omit that proof here, since the result is an immediate consequence of case (\ref{Thm:case-3}) of Theorem \ref{Thm:determ} below.

\subsection{Main results}
We now consider a natural extension of the simple model of a bi-infinite sequence of perfect reflector coins in $U=SC$. Instead of perfect reflectors, we allow \emph{imperfect} reflectors that approach a perfect reflector at infinity. Specifically, we assume the existence of a bi-infinite subsequence $(a(j_m))$ satisfying $a(j_m)\to 0$ as $|m|\to\infty$, and investigate when the maximal velocity \eqref{def:v} vanishes.

We note that  all of our results extend without change to the CMV setting. More precisely, suppose that the corresponding CMV representation is determined by the sequence of Verblunsky pairs $((\alpha_n,\rho_n))_{n\in\Z}$ as in \eqref{def:Verblunsky}. Then the conclusions of our results remain valid upon identifying $|a_1(n)|=|\rho_{2n}|$ and $|a_2(n)|=|\rho_{2n+1}|$ for all $n\in\Z$.

Our main theorem establishes vanishing of the maximal velocity under appropriate conditions on the growth of the gaps between successive $j_m$ 
\begin{equation}\label{def:g}
g_m :=
\begin{cases}
j_{m+1}-j_m, & m\geq 0 \\
j_m - j_{m-1}, & m<0,
\end{cases}
\end{equation}
and/or on the convergence rate of the reflectors.
\begin{theorem}\label{Thm:determ}
Let $W=S_+C_1S_-C_2$ be the split-step quantum walk associated with coins
$C_1$ and $C_2$ associated with $(a_1(n))_{n\in\Z}$ and $(a_2(n))_{n\in\Z}$ as in \eqref{def:a}.
Fix $k\in\{1,2\}$ and assume that there exists a strictly increasing bi-infinite sequence
$(j_m)_{m\in\Z}\subset\Z$ (cf.\ \eqref{def:j}) with $j_m\to\pm\infty$ as $m\to\pm\infty$ such that
\begin{equation}\label{cond-0}
|a_k(j_m)|\xrightarrow[|m|\to\infty]{}0.
\end{equation}
 If one of the following holds,
then the maximal velocity $v(W)$ defined in \eqref{def:v} vanishes, i.e.\ $v(W)=0$:
\begin{enumerate}[(i)]\itemsep0.25em

\item \label{Thm:case-1}\emph{Uniformly bounded gaps:}
\begin{equation}\label{cond-uniform-bound}
\sup_{m\in\Z} g_m < \infty .
\end{equation}

\item \label{Thm:case-2}\emph{Sublinear gaps with quantitative decay:}
\begin{equation}\label{cond-a-decay}
\quad
\frac{g_m}{|j_m|}\xrightarrow[|m|\to\infty]{}0  \quad\text{\ and \ }\quad |a_k(j_m)| = \mathcal{O}(|j_m|^{-1}).
\end{equation}

\item \label{Thm:case-3}\emph{Gap-weighted decay:}
\begin{equation}\label{cond-gap-weighted}
\max\{|j_m|,g_m\}\,|a_k(j_m)| \xrightarrow[|m|\to\infty]{} 0.
\end{equation}

\end{enumerate}
\end{theorem}
The proofs of cases (\ref{Thm:case-1}) and (\ref{Thm:case-2}) are based on the general bound established in Theorem \ref{Thm:general} and are presented in Section \ref{sec:zero-velocity-1-2}. The proof of case (\ref{Thm:case-3}) is independent of the arguments for cases (\ref{Thm:case-1}) and (\ref{Thm:case-2}) and is given in Section \ref{sec:pf:case3}.

We first make some general remarks on Theorem \ref{Thm:determ}, and then discuss each of the three cases in the theorem separately.
\begin{enumerate}[(A)]\itemsep3mm

\item The conclusion $v(W)=0$ should be interpreted as the absence of ballistic transport. This agrees with \cite[Def.~2.1]{damanikWhatBallisticTransport2024}, where a state $\psi$ is said to exhibit \emph{strong ballistic transport} when its maximal velocity is nonzero, and to exhibit \emph{ballistic transport in the norm-growth sense} when there exist constants $c,C>0$ such that
$ct \le \|QW^t\psi\| \le Ct$. 

Moreover, for any initial state with finite second moment, the probability of finding the
walker at a distance of order  $t$ from the
origin tends to zero as $t\to\infty$, for every initial state in $D(Q)$. Thus the
theorem excludes linear spreading of the quantum walk, although slower forms of transport
may still occur.
To make this precise, recall that the square of the velocity $v(W)$ can be written as
\begin{equation}\label{eq:velocity-moment}
(v(W))^2
=\sup_{\substack{\psi\in D(Q)\\ \|\psi\|=1}}
\limsup_{t\to\infty}\frac{1}{t^{2}}
\sum_{j\in\mathbb{Z}}\sum_{k=\pm} j^{2}\bigl|\langle \delta_j^k \mid W^{t}\psi\rangle\bigr|^{2}.
\end{equation}
For a fixed normalized vector $\psi\in D(Q)$ and $t\in\mathbb{N}$, define
\begin{equation}\label{eq:p-psi-t}
p_{\psi,t}(j):=\sum_{k=\pm}\bigl|\langle \delta_j^k \mid W^{t}\psi\rangle\bigr|^{2},
\qquad j\in\mathbb{Z}.
\end{equation}
Since $\sum_{j\in\mathbb{Z}}p_{\psi,t}(j)=1$, the function $p_{\psi,t}$ defines a
probability measure $\mathbb{P}_{\psi,t}$ on $\mathbb{Z}$ by $
\mathbb{P}_{\psi,t}(\{j\})=p_{\psi,t}(j)$.
We interpret $p_{\psi,t}(j)$ as the probability of finding the walker at site $j$ after
$t$ steps, given the initial state $\psi$.

Let $\mathcal{J}$ be the random variable on
$(\mathbb{Z},\mathcal{P}(\mathbb{Z}),\mathbb{P}_{\psi,t})$ defined by $\mathcal{J}(j)=j$, and let
$\mathbb{E}_{\psi,t}$ denote expectation with respect to $\mathbb{P}_{\psi,t}$. Then
\eqref{eq:velocity-moment} can be rewritten as
\begin{equation}\label{eq:vU-prob}
(v(W))^2
=\sup_{\substack{\psi\in D(Q)\\ \|\psi\|=1}}
\limsup_{t\to\infty}\mathbb{E}_{\psi,t}\!\left(\left(\mathcal{J}/t\right)^2\right).
\end{equation}

In particular, if $v(W)=0$, then for every normalized $\psi\in D(Q)$,
$\mathbb{E}_{\psi,t}\!\left(\left(\mathcal{J}/t\right)^2\right)\to 0$ as $t\to\infty$.
Hence, for every $v>0$, Chebyshev's inequality gives
\begin{equation}
\mathbb{P}_{\psi,t}\bigl(|\mathcal{J}|\geq vt\bigr)
\leq\frac{\mathbb{E}_{\psi,t}(\mathcal{J}^2)}{v^2t^2}
= \frac{1}{v^2}\,\mathbb{E}_{\psi,t}\!\left(\left(\mathcal{J}/t\right)^2\right)
\longrightarrow 0
\qquad \text{as } t\to\infty.
\end{equation}
Thus, asymptotically, the walker cannot be found with positive probability at a linear
distance $vt$ from the origin. This is precisely the absence of ballistic transport.
Of course, this does not rule out slower spreading, for instance transport of order
$t^\beta$ with $0\leq\beta<1$.
\item Theorem \ref{Thm:determ} provides general scenarios of absence of ballistic transport, depending solely on the existence of a bi-infinite sequence $(j_m)_{m\in\Z}$ along which the relevant coin entries vanish at infinity.
In particular, the conclusion is independent of the values of the coins on $\Z\setminus\{j_m:\ m\in\Z\}$.

The generality of Theorem \ref{Thm:determ} with respect to the coin values outside the distinguished index set $\{j_m:\,m\in\Z\}$ shows that the resulting bound is insensitive to possible interference effects generated by specific configurations of coins away from that set. In this sense, the criterion captures a mechanism for the absence of ballistic transport that depends only on suitably sparse local information and not on global oscillatory or resonant features of the full coin sequence.

 It is also worth emphasizing that our method yields bounds that do not take into account possible combined effects of the two coin sequences defining $C_1$ and $C_2$. That is, the argument treats the relevant decay information in a decoupled manner and therefore does not exploit cancellations or other interaction phenomena that may arise from the joint structure of $C_1$ and $C_2$.

In particular, Theorem \ref{Thm:determ} does not yield a conclusion for quantum-walk models whose inhomogeneity is generated by phases, either random or quasi-periodic. For example, \cite{CFGW2020LMP} proves singular continuity of the spectrum for magnetic split-step quantum walks, where the inhomogeneity is induced by quasi-periodic phases. Similarly, in \cite{JM10}, random inhomogeneity is implemented through i.i.d. phases. In both cases, however, the moduli of the associated transmission parameters are constant.

\item Theorem \ref{Thm:determ} gives sufficient conditions for the absence of ballistic transport, but it does not resolve the finer dynamical behavior in the cases covered by \eqref{Thm:case-1}--\eqref{Thm:case-3}. This naturally leads to the question of whether these hypotheses give rise to localization, subballistic transport, or some intermediate regime. It is also natural to ask for the corresponding spectral picture: which spectral types can occur under the assumptions of Theorem \ref{Thm:determ}, and how are they tied to the absence of ballistic transport? The RAGE theorem for unitary operators, see, e.g., \cite{KumarSabri2026}, offers a general mechanism connecting spectral decompositions with dynamical spreading, whereas \cite{CedzichFillmanVelazquez2026} shows that pure point spectrum forces zero maximal velocity. In view of this, a particularly interesting question is whether there exists a quantum walk with vanishing maximal velocity and purely absolutely continuous spectrum.
\end{enumerate}

Here are some remarks about the three cases in Theorem \ref{Thm:determ}. To simplify the discussion, we restrict attention to the
positive branch $j_m$, $m\geq 0$; the corresponding statements for the negative
branch $j_{-m}$ are entirely analogous.

\begin{enumerate}[]\itemsep3mm

\item \underline{Case (\ref{Thm:case-1}):} the gaps $g_m$ are uniformly bounded, so the almost
reflecting sites occur with bounded spacing. No quantitative decay rate is imposed on
the convergence $|a_k(j_m)|\to 0$: the mere presence of infinitely many such
sites at uniformly bounded distances is already sufficient to suppress
ballistic transport and hence to force $v(W)=0$.

\item \underline{Case (\ref{Thm:case-2}):} the bounded-spacing hypothesis is relaxed. The gaps
$g_m$ are allowed to diverge, but only sublinearly relative to $|j_m|$, so that
the spacing between successive almost reflecting sites remains negligible on the
scale of their distance from the origin. This weaker spacing assumption is
offset by the quantitative decay condition
$|a_k(j_m)|=\mathcal{O}(|j_m|^{-1})$, which ensures that the corresponding
sites still act as effective reflectors at large scales. The condition
$g_m/|j_m|\to 0$ is satisfied, for instance, when $j_m\sim m^\alpha$ for
$\alpha>0$, or more generally when $j_m\sim \exp(m^\beta)$ with
$\beta\in(0,1)$. By contrast, it fails for exponentially growing indices, for
example when $j_m\sim 2^m$, since in that case
\begin{equation}
\frac{g_m}{|j_m|}
\sim \frac{2^{m+1}-2^m}{2^m}=1.
\end{equation}

\item  \underline{Case (\ref{Thm:case-3}):} any a priori restriction on the growth of the
gaps is removed. In particular, the sequence $(j_m)$ may be exponentially sparse, or even
sparser. This additional sparsity is compensated by the stronger decay
assumption \eqref{cond-gap-weighted}.
Thus the relevant smallness of $a_k(j_m)$ is measured simultaneously against
 $j_m$ and against the local spacing to the next almost
reflecting site. 
When the gaps are large, the condition
$g_m|a_k(j_m)|\to 0$ becomes the decisive requirement: the greater the
separation between successive almost reflecting sites, the faster the
corresponding coin parameters must decay. For example, if $j_m\sim 2^m$, then
$g_m\sim 2^m$, so \eqref{cond-gap-weighted} is satisfied whenever
$|a_k(j_m)|=o(2^{-m})$. We also note that \eqref{cond-gap-weighted} is equivalent to
\begin{equation}
j_{m+1}|a_k(j_m)|\xrightarrow[|m|\to\infty]{} 0.
\end{equation}
Case (\ref{Thm:case-3}) allows arbitrary
gap growth, at the price of requiring a corresponding increase in the
reflecting strength along the chosen subsequence.

A particularly relevant special case of \eqref{Thm:case-3} occurs when the relative gaps satisfy
\begin{equation}\label{eq:case-2.5}
\limsup_{|m|\to\infty}\frac{g_m}{|j_m|}\in(0,\infty)   \quad\text{\ and \ }\quad |a_k(j_m)| = o(|j_m|^{-1}).
\end{equation}
Then $g_m\sim |j_m|$ as $m\to\infty$, so the two conditions appearing in \eqref{cond-gap-weighted} are equivalent:
\begin{equation}
|j_m|\,|a_k(j_m)|\to 0
\qquad\Longleftrightarrow\qquad
g_m\,|a_k(j_m)|\to 0.
\end{equation}
We note here that $v(W)=0$ under \eqref{eq:case-2.5} follows also from Theorem \ref{Thm:general} below, see Remark \ref{rem:case-2.5}.

Among the three cases, this is the regime most closely related to the sparse growing barrier picture familiar from the CMV setting studied by Damanik et al.\ \cite{DamanikEtAl2016}, see the next remark. 

\end{enumerate}

\begin{remark}\label{ref:DamanikSparce}
A related setting was studied by Damanik et al.~\cite{DamanikEtAl2016}. Their model is formulated in terms of the Verblunsky coefficients in a CMV setting that is equivalent to a shift-coin quantum walk $SC$ on the half-line in the present framework. In their notation, the sparse reflector sites are denoted by $L_m$; these correspond to the sites $j_m$ in our notation.
 They consider  reflectors placed at sites $L_m$ that grow extremely rapidly with $m$. More precisely, they assume
\begin{equation}
\frac{\log(L_1L_2\cdots L_{N-1})}{\log(L_N)}\xrightarrow[N\to\infty]{}0
\qquad\Longleftrightarrow\qquad
\frac{\log(L_m)}{\log(L_{m-1})}\xrightarrow[m\to\infty]{}\infty.
\end{equation}
In particular, $L_{m+1}$ grows faster than any fixed power of $L_m$. Typical examples are
$
L_m=2^{m^m}$
 and $L_m=2^{m!}$. These are compatible with case (\ref{Thm:case-3}) in Theorem \ref{Thm:determ}.
The coins are chosen so that
\begin{equation}
|a(L_m)|=L_m^{-\frac{1-\eta}{2\eta}}
\text{ where }\eta\in(0,1),
\text{ and } 
C(n)=\idty_2 \text{ for } n\notin\{L_m : m\in\Z_+\},
\end{equation}
so that the walk is perfectly transmitting away from the sparse reflector sites.

One of the main results of \cite{DamanikEtAl2016} is that the upper and lower transport exponents do not coincide; moreover, the upper transport exponent attains the ballistic value $1$, hence the system exhibit \emph{quantum intermittency}. Thus, the propagation is highly nonuniform across different time scales. However, these results do not directly determine the maximal velocity. On the other hand, our Theorem~\ref{Thm:determ} does not imply that the maximal velocity vanishes in this regime either, since
\begin{equation}
L_{m+1}|a(L_m)|
=
L_{m+1}L_m^{-\frac{1-\eta}{2\eta}}
\longrightarrow\infty
\end{equation}
for every $\eta\in(0,1)$. Therefore, neither \cite{DamanikEtAl2016} nor Theorem~\ref{Thm:determ} settles whether the maximal velocity vanishes in this setting. Since our notion of maximal velocity is also defined through a $\limsup$ in time, the fact that the upper transport exponent equals $1$ is at least compatible with positive maximal velocity in our sense. We note here that if $|a(L_m)|$ are assumed to have the stronger decay of $o(L_{m+1})$ then we guarantee zero-velocity by part (\ref{Thm:case-3}) in Theorem \ref{Thm:determ} regardless of the values of $|a(n)|$ for $n\notin\{L_m, m\in\Z_+\}$.

\end{remark}


Our methods yield a more general \emph{a priori} bound on the maximal velocity, from which cases (\ref{Thm:case-1}) and (\ref{Thm:case-2}) in Theorem \ref{Thm:determ} follow as a special case.
\begin{theorem}\label{Thm:general}
Let $W=S_+C_1S_-C_2$ be the split-step quantum walk, where the coins $C_1$ and $C_2$ are associated with the sequences $(a_1(n))_{n\in\Z}$ and $(a_2(n))_{n\in\Z}$ as in \eqref{def:a}. Then the maximal velocity $v(W)$ defined in \eqref{def:v} satisfies
\begin{equation}\label{eq:v-bound-general}
v(W)\leq \min_{k\in\{1,2\}}\inf_{(j_m)\in J}\inf_{N\in\N}\left(1+q_{(j_m)})\right)^{N+1}
\max\left\{
\sup_{m\geq N} g_{m-N}|a_k(j_{m+1})|,
\sup_{m\geq N} g_{-m+N}|a_k(j_{-m})|
\right\},
\end{equation}
where the first infimum is taken over
\begin{equation}\label{def:J}
J:=\left\{(j_m)_{m\in\Z}\subset\Z:\ j_m<j_{m+1}\ \forall m\in\Z
\ \text{and}\ 
\max_{\pm}\limsup_{m\to\pm\infty}\frac{g_m}{j_m}<\infty\right\}.
\end{equation}
In \eqref{eq:v-bound-general},
\begin{equation}
q_{(j_m)}:=\max_{\pm}\limsup_{m\to\pm\infty}\frac{g_m}{j_m},
\end{equation}
and $(g_m)_{m\in\Z}$ denotes the gap sequence defined in \eqref{def:g}.
\end{theorem}
The proof of Theorem \ref{Thm:general} is presented in Sections \ref{sec:properties} and \ref{sec:pf-general}.

In the simple case where $a_1=a_2=a$, and $|a(m)|=|a(-m)|$ for all $m\in\Z$, we obtain the bound
\begin{equation}
v(W)\leq \inf_{(j_m)\in J}\inf_{N\in\N}\left(1+q_{(j_m)}\right)^{N+1}
\sup_{m\geq N} g_{m-N}|a(j_{m+1})|.
\end{equation}

The usefulness of Theorem \ref{Thm:general} becomes particularly transparent when the gaps are uniformly bounded. In that case the weights $g_m$ can be pulled out of the tail supremum, and the estimate reduces to a bound in terms of the two-sided limit superior of a subsequence of the coin parameters.

\begin{coro}\label{Cor:uniform-gap-bound}
Assume that $\sup_m g_m<\infty$. Then
\begin{equation}\label{v-bound-uniform-1}
v(W)\leq \min_{k\in\{1,2\}}
\inf_{\substack{(j_m)\in J\\ \sup_m g_m<\infty}}
\big(\sup_{m\in\Z} g_m\big)
\max\left\{
\limsup_{m\to\infty}|a_k(j_m)|,\,
\limsup_{m\to-\infty}|a_k(j_m)|
\right\}.
\end{equation}
\end{coro}

In particular, under the additional assumption of uniformly bounded gaps, the maximal velocity can be controlled solely through the two-sided limit superior of a subsequence of $(|a_k(n)|)_{n\in\Z}$. Thus the resulting upper bound depends only on the values of the coin parameters along the selected subsequence $(j_m)_{m\in\Z}$, and is independent of the values on its complement $\Z\setminus\{j_m:\ m\in\Z\}$.

A particularly useful instance of \eqref{v-bound-uniform-1} is obtained by choosing the trivial subsequence $j_m=m$ for all $m\in\Z$. This gives the general non-trivial \emph{a priori} bound
\begin{equation}\label{v:apriori}
v(W)\leq \min_{k\in\{1,2\}}
\max\left\{
\limsup_{n\to\infty}|a_k(n)|,\,
\limsup_{n\to-\infty}|a_k(n)|
\right\}.
\end{equation}
In the spatially homogeneous case with transmission parameter $a$ \footnote{The $(1,1)$-entry, $a$, of the coin operator  is understood as the \emph{transmission parameter}, since in this setting the velocity equals $|a|$; see, e.g., \cite{ARS2023-CMP, ARCSW25}.}, it is known that the velocity equals $|a|$, see \cite{ARCSW25}. The bound \eqref{v:apriori} shows that the upper bound $v(W)\leq |a|$ persists under arbitrary perturbations of the coins that are asymptotically vanishing at both $+\infty$ and $-\infty$; see, e.g., \cite{MaedaSuzukiWada2022,FudaFunakawaSuzuki2018,MaedaSasakiSegawaSuzukiSuzuki2022} for short- and long-range perturbations of the coin operators.

From the physical point of view, Theorem \ref{Thm:general} shows that the asymptotic transport is governed not by the full coin profile, but by the existence of sufficiently weak scattering sites distributed throughout the lattice with controlled spacing. This is precisely the mechanism behind the zero-velocity criteria established in Theorem \ref{Thm:determ} that follows by choosing subsequences for which the relative gaps (to position) vanish, see Section \ref{sec:zero-velocity-1-2}.



Next, we turn to a random setting and ask how the presence of randomly placed
 sites influences the large-scale dynamics of the quantum
walk. In contrast to the deterministic setting treated in
Theorem \ref{Thm:determ}, the coins are now generated by a sequence of
independent identically distributed unitary-$2\times 2$-valued random variables.  The following result shows that, under a mild
assumption on the distribution of the coin entries, the maximal velocity
vanishes almost surely.

In the sequel of this section we assume that $(a_k(n))_{n\in\Z}$ (for $k=1,2$) is a sequence of
i.i.d. random variables, and we
denote by $F$ the distribution function of $|a_k(n)|$ given as
  $F(x) = \mathbb{P}(|a_k(0)|\leq x)$, for $x\geq 0$.

\begin{coro}\label{coro:random}
Consider the split-step quantum walk $W=S_+ C_1 S_- C_2$ where $C_k$, for $k=1,2$ is associated with the bi-infinite sequence $(a_k(n))_{n\in\Z}$.
Fixed $k\in\{1,2\}$ and assume that $(a_k(n))_{n\in\Z}$ are i.i.d.\ random variables, and let $F$ be the
distribution function of $|a_k(0)|$. Suppose there exist constants $c>0$,
$\alpha\in(0,1)$, and $x_0>0$ such that
\begin{equation}\label{thm:eq:F}
  F(x) \geq c\, x^\alpha
  \qquad \text{for all } x\in(0,x_0).
\end{equation}
Then the maximal velocity $v(W)$ vanishes, i.e. $v(W) = 0$, almost surely.
\end{coro}

The random setting should be viewed as an application of the general method rather
than as a sharp random-model result in its own right. Indeed, for many disordered
models the literature establishes stronger conclusions, such as dynamical
localization or even Anderson localization. The point of
Corollary~\ref{coro:random} is different: it is a direct consequence of our general
a priori bound together with the deterministic zero-velocity criterion in Theorem \ref{Thm:determ}, and it
isolates a simple probabilistic mechanism for vanishing velocity through the lower
tail of the distribution of $|a_k(n)|$ near zero. In this sense, the novelty lies
not in obtaining a strong random result, but in the fact that the conclusion
follows immediately from the deterministic framework. For related rigorous results
on random CMV matrices, we mention in particular \cite{Zhu2024}, where Anderson
localization and dynamical localization are proved for arbitrary
i.i.d.\ Verblunsky coefficients.

 Here are some remarks
\begin{enumerate}[(A)]\itemsep3mm
\item  
Recall that sites with $a_k(n)\approx 0$ act as (almost) perfect reflectors for
the walk. The deterministic result in Theorem \ref{Thm:determ} shows that a
sufficiently dense family of such reflectors forces the maximal velocity
to vanish. Corollary \ref{coro:random} states that, in a random medium where
the coin entries $(a_k(n))$ are i.i.d.\ and small values of $|a_k(n)|$ occur often
enough, this situation arises with probability one: the walk cannot propagate
ballistically, and maximal velocity
is almost surely zero.

\item 
The lower bound \eqref{thm:eq:F}
for small $x>0$ means that the probability of finding a very small $|a_k(0)|$
is at least of order $x^\alpha$. Since $\alpha<1$, this grows \emph{faster}
than linearly as $x\to 0$, so the distribution has a significant amount of
mass near $0$: very small $|a_k(n)|$ occur relatively frequently.

The condition \eqref{thm:eq:F} is satisfied in a variety of
natural situations:
\begin{enumerate}[(a)]
\item If $|a_0|$ has density
\begin{equation}
  f(x) = \alpha x^{\alpha-1}\,\idty_{(0,1)}(x),
\end{equation}
then $F(x) = x^\alpha$ for $x\in(0,1)$, and the hypothesis holds with $c=1$.

\item If $\mathbb{P}(|a_k(0)|=0)>0$, then $F(x)\geq \mathbb{P}(|a_k(0)|=0)$ for all
sufficiently small $x>0$, so the condition $F(x)\geq c x^\alpha$ holds
for some choice of $c>0$ and any $\alpha\in(0,1)$ (after possibly reducing
$c$ and restricting to $x\in(0,x_0)$).
\end{enumerate}

\item Most rigorous results on disordered quantum walks concern \emph{static} spatial
disorder, and are predominantly formulated in the line of i.i.d.\ random phase applied to coins directly or through a random field, see, e.g.,
\cite{BHJ2003,HamzaJoyeStolz2006,hamzaLyapunovExponentsUnitary2007,JM10, HamzaJoyeStolz2009, Konno2009QWRE}, for one-dimensional results, as well
as higher-dimensional  extensions
\cite{Joye2012,HJ14,AschJoye2019,HJS2009,joye2026dynamical}. For instance, in the one-dimensional model of Joye--Merkli \cite{JM10}, the disorder
is implemented through i.i.d.\ random phases, while the moduli of the relevant transition amplitudes remain fixed.

A complementary one-dimensional random-coin model was studied in
\cite{AhlbrechtVolkherScholzWerner2011}, where dynamical and strong spectral
localization were proved under suitable assumptions on the associated transfer
matrices. By contrast, Corollary~\ref{coro:random} is sensitive to the lower tail
of the distribution of $|a_k(n)|$ at the origin: it yields a vanishing-velocity
criterion from the condition
$F(x)\ge c x^\alpha$
for sufficiently small $x$. To the best of our knowledge, this type of
zero-velocity statement, formulated directly in terms of the probability of very
small transmission amplitudes, does not seem to follow from the existing rigorous
localization results in the literature.

\end{enumerate}

\section{Some properties of the maximal velocity}\label{sec:properties}

In this section we establish two general properties of the maximal velocity $v(W)$ defined in  \eqref{def:v}. First, we show that $v(W)$ is stable under perturbations of the position operator $Q$ defined in \eqref{def:Q}: if $\tilde Q$ differs from $Q$ by an operator that is $Q$-bounded with relative bound less than 1, then the associated velocity is proportional to $v(W)$. This provides a robustness statement showing that the maximal velocity does not depend sensitively on the particular representative of the position operator.

We also prove a simple but useful symmetry property for products of unitary operators. Namely, when $U_1$ and $U_2$ have bounded commutators with $Q$, the products $U_1U_2$ and $U_2U_1$ have the same maximal velocity. This invariance under cyclic exchange of two factors is crucial in observing the symmetry of our results with respect to the parameters that determine $C_1$ or $C_2$ (or $\CL$ or $\CM$) as $v(W)=v(\CL\CM)=v(\CM\CL)$.

\begin{theorem}\label{Thm:Q}
Let $\tilde{Q}$ be an operator on $\mathcal H$ such
that $(Q-\tilde{Q})$ is $Q$-bounded 
\footnote{Recall that $(Q-\tilde{Q})$ being $Q$-bounded with relative bound $b<1$ means that, $D(Q)\subset D(Q-\tilde{Q})$ and for any $\epsilon>0$ there exists $b_\epsilon\in(b,b+\epsilon)$ and $c_\epsilon\geq 0$
such
that
\begin{equation}
  \|(Q-\tilde{Q})\phi\|
  \leq c_\epsilon\|\phi\| + b_\epsilon\|Q\phi\|
  \qquad\text{for all }\phi\in D(Q).
\end{equation}
In particular, $\tilde Q\phi = Q\phi - (Q-\tilde Q)\phi$ is well-defined for
every $\phi\in D(Q)$, see, e.g., \cite[Page 93]{Weidmann1980}.
}
 with relative bound $b<1$. Then
\begin{equation}\label{lem:eq:v}
 (1-b) v(W)
  \leq \sup_{\substack{\psi\in D(Q)\\\|\psi\|=1}}
    \limsup_{t\to\infty}\frac{1}{t}\bigl\|\tilde{Q}\, W^t \psi\bigr\| \leq (1+b)v(W).
\end{equation}
\end{theorem}
In particular, if $\tilde{Q}-Q$ is $Q$-bounded with relative bound zero than $\tilde{Q}$ can replace $Q$ in the definition of the maximal velocity, i.e.,
\begin{equation}
b=0 \ \ \Rightarrow \ \ v(W)=\sup_{\substack{\psi\in D(Q)\\\|\psi\|=1}} \limsup_{t\to\infty}\frac{1}{t}\bigl\|\tilde{Q}\, W^t \psi\bigr\|.
\end{equation}

\begin{proof}[Proof of Theorem \ref{Thm:Q}]
Fix $\psi\in D(Q)$ with $\|\psi\|=1$ and set
\begin{equation}
  \zeta_t := \frac{1}{t}\,\|Q W^t\psi\|, \qquad
  \eta_t := \frac{1}{t}\,\|\tilde{Q} W^t\psi\|, \qquad t\in\N.
\end{equation}
By the triangle inequality we have
\begin{equation}\label{pf:lem:triangle}
  |\zeta_t-\eta_t|
  = \frac{1}{t}\bigl|\|Q W^t\psi\| - \|\tilde Q W^t\psi\|\bigr|
  \leq \frac{1}{t}\,\|(Q-\tilde Q)W^t\psi\|.
\end{equation}
Since $(Q-\tilde{Q})$ is $Q$-bounded with relative bound $b<1$, then given any $\epsilon>0$, there exists $b_\epsilon\in (b,b+\epsilon)$ and $c_\epsilon\geq 0$
such that
\begin{equation}
  \|(Q-\tilde Q)\phi\|
  \leq b_\epsilon\|Q\phi\| + c_\epsilon\|\phi\|
  \qquad\text{for all }\phi\in D(Q).
\end{equation}
Applying this with $\phi = W^t\psi$ in \eqref{pf:lem:triangle} and recalling that $W$ is unitary gives
\begin{equation}
  |\zeta_t-\eta_t|
  \leq b_\epsilon \zeta_t + \frac{c_\epsilon}{t}\|\psi\|
  = b_\epsilon \zeta_t + \frac{c_\epsilon}{t},\qquad t\in\N.
\end{equation}
Hence
\begin{equation}
  (1-b_\epsilon)\zeta_t - \frac{c_\epsilon}{t}
  \leq \eta_t
  \leq (1+b_\epsilon)\zeta_t + \frac{c_\epsilon}{t}.
\end{equation}
Taking the limit superior as $t\to\infty$ and using that $\zeta_t,\eta_t\geq 0$ then take the limit as $\epsilon\to 0$, noting that $b_\epsilon\to b$, we obtain
\begin{equation}
  (1-b)\,\limsup_{t\to\infty} \zeta_t
  \leq \limsup_{t\to\infty} \eta_t
  \leq (1+b)\ \limsup_{t\to\infty} \zeta_t
\end{equation}
where these limsups take values in $[0,\infty]$.

Because $\psi\in D(Q)$ with $\|\psi\|=1$ is arbitrary, taking the supremum over all such $\psi$ yields
\eqref{lem:eq:v}, and the lemma follows.
\end{proof}

\begin{lemma}\label{lem:com-velocity}
Let $U_1,U_2$ be unitary operators such that
$[Q,U_k]$ is a bounded operator on $\CH$ for $k=1,2$.
Then $[Q,U_1U_2]$ and $[Q,U_2U_1]$ are bounded as well, and
\begin{equation}
v(U_1U_2)=v(U_2U_1).
\end{equation}
\end{lemma}

\begin{proof}
For $\varphi\in D(Q)$ we have
$
Q\,U_k\varphi \;=\; U_k\,Q\varphi + [ Q,U_k]\varphi$,  for $k=1,2$.
Hence $U_k\varphi\in D(Q)$, so 
\begin{equation}\label{eq:subset1}
U_kD(Q)\subset D(Q).
\end{equation}
Moreover,
$[Q,U_k^*] = -\,U_k^*[Q,U_k]U_k^*$ is bounded on $\mathcal{H}$,
so the same argument gives $U_k^*D(Q)\subset D(Q)$, hence 
\begin{equation}\label{eq:subset2}
D(Q)\subset U_k D(Q).
\end{equation}
\eqref{eq:subset1} and \eqref{eq:subset2} give
\begin{equation}\label{Domain}
U_kD(Q)=D(Q)\quad\text{and}\quad U_k^*D(Q)=D(Q)\qquad (k=1,2).
\end{equation}
In particular, $(U_1U_2)^t$ and $(U_2U_1)^t$ preserve $D(Q)$ for all $t\in\mathbb N$.

Fix $\psi\in D(Q)$ with $\|\psi\|=1$ and $t\in\mathbb N$. Use $
(U_1U_2)^t \;=\; U_1 (U_2U_1)^t U_1^*$ and
define $\phi_t := (U_2U_1)^t U_1^*\psi$. By \eqref{Domain}, $\phi_t\in D(Q)$, and $\|\phi_t\|=\|\psi\|=1$.
Then
\begin{align}
\|Q(U_1U_2)^t\psi\|
&= \|Q\,U_1\phi_t\|
 = \bigr\|U_1Q\phi_t + [Q,U_1]\phi_t \bigr\| \notag \\
&\leq \|Q\phi_t\| + \bigr\|[Q,U_1]\bigr\|\,\|\phi_t\|
= \|Q (U_2U_1)^t U_1^*\psi\| + \bigr\|[Q,U_1]\bigr\|.
\end{align}
Divide by $t$ and take $\limsup$ as $t\to\infty$ to obtain
\begin{equation}
\limsup_{t\to\infty}\frac1t\|Q(U_1U_2)^t\psi\|
\le
\limsup_{t\to\infty}\frac1t\|Q (U_2U_1)^t U_1^*\psi\|.
\end{equation}
Now take the supremum over $\psi\in D(Q)$ with $\|\psi\|=1$. Since $U_1^*$ is unitary and maps $D(Q)$ onto itself \eqref{Domain},
then
\begin{equation}
v(U_1U_2)\le
\sup_{\substack{\eta\in D(Q)\\ \|\eta\|=1}}
\limsup_{t\to\infty}\frac1t\|Q (U_2U_1)^t \eta\|
= v(U_2U_1).
\end{equation}
Interchanging the roles of $U_1$ and $U_2$ yields the reverse inequality $v(U_2U_1)\leq v(U_1U_2)$, and therefore
$v(U_1U_2)=v(U_2U_1)$.
\end{proof}

\section{A general a priori bound for the maximal velocity}\label{sec:pf-general}

In this section we prove Theorem \ref{Thm:general}. Our starting point is the natural decomposition of $\ell^2(\mathbb{Z})$ induced by a strictly increasing sequence $(j_m)_{m\in\mathbb{Z}}$ satisfying
\begin{equation}
\limsup_{|m|\to\infty}\frac{g_m}{|j_m|}<\infty.
\end{equation}
Relative to this decomposition, we introduce a modified position operator $\tilde Q$, which, by Theorem \ref{Thm:Q}, may be used in place of $Q$ in the definition of the maximal velocity.

Motivated by the perfect reflectors setting discussed in Section \ref{sec:motivation}, for a given sequence of indices $(j_m)_{m\in\mathbb{Z}}\in J$, with $J$ as in \eqref{def:J}, we decompose $\ell^2(\mathbb{Z})$ as in \eqref{dec}. More precisely, for each $m\in\mathbb{Z}$, we define
\begin{equation}\label{def:Hm}
  \mathcal{H}_m
  = \Span\{\delta_n : n = 2j_m,2j_m+1,\dots,2j_{m+1}-1\}.
\end{equation}
Then $\ell^2(\mathbb{Z})$ admits the orthogonal decomposition (see Figure \ref{fig:dec})
\begin{equation}\label{decomposition}
  \ell^2(\mathbb{Z}) = \bigoplus_{m\in\mathbb{Z}} \mathcal{H}_m.
\end{equation}

Recall the correspondence $W=S_+ C_1 S_- C_2 \equiv \CL\CM$ between the split-step walk and the extended CMV matrices defined in \eqref{def:CMV}. The operator $\CL$ is block-diagonal, with $2\times 2$ blocks $\sigma_x C_1(n)$ relative to the decomposition
\begin{equation}
\big\{\Span\{\delta_{2n},\delta_{2n+1}\}\big\}_{n\in\mathbb{Z}},
\end{equation}
which clearly refines $\{\mathcal{H}_m\}_{m\in\mathbb{Z}}$. Hence $\CL$ is also block-diagonal with respect to \eqref{decomposition}.

To describe $\CM$ relative to \eqref{decomposition}, it is convenient to relabel (and order) the canonical basis of $\mathcal{H}_m$ as
\begin{equation}\label{def:e-m-j}
e_{m|r}:=\delta_{2j_m+r-1} \text{ for }r=1, \ldots, d_m,
\qquad
d_m:=\dim(\mathcal{H}_m)=2|j_{m+1}-j_m|.
\end{equation}
Let $P_m$ denote the orthogonal projection onto $\mathcal{H}_m$. Then $\CM$ takes the block tridiagonal form
\begin{equation}\label{def-CM-dec}
\CM
= \sum_{m\in\mathbb{Z}} P_m \CM P_m
+ \sum_{m\in\mathbb{Z}}
\left(
\overline{a_2(j_{m+1})}\,|e_{m|d_m}\rangle\langle e_{m+1|1}|
+
a_2(j_{m+1})\,|e_{m+1|1}\rangle\langle e_{m|d_m}|
\right).
\end{equation}
The first sum consists of the diagonal $d_m\times d_m$ blocks of $\CM$ and will play no role in our analysis, since it commutes with the modified position operator introduced in \eqref{def:tilde-Q-N} below.

With respect to \eqref{decomposition}, the position operator decomposes as
\begin{equation}
Q \equiv \bigoplus_{m\in\mathbb{Z}} \restr{Q}{\mathcal{H}_m},
\end{equation}
where
\begin{equation}\label{def:Q-Hm}
\restr{Q}{\mathcal{H}_m}
=
\begin{bmatrix}
j_m & & & & \\
& \boxed{\tiny\begin{matrix} j_m+1 & \\ & j_m+1 \end{matrix}} & & & \\
& & \ddots & & \\
& & & & j_{m+1}
\end{bmatrix}_{d_m\times d_m}.
\end{equation}
A key ingredient in the argument below is Theorem \ref{Thm:Q}, which allows us to replace $Q$ by a suitable block-scalar operator $\tilde Q$ with respect to the decomposition \eqref{decomposition}. The latter is chosen so that $[\tilde Q,\CL]=0$, and therefore
\begin{equation}
[\tilde Q,\CL\CM]=\CL[\tilde Q,\CM].
\end{equation}
In this way, $\CM$ (equivalently, $C_2$) plays the role of the \emph{active} operator, leading to the bound \eqref{eq:v-bound} below in terms of the sequence $(a_2(n))_{n\in\mathbb{Z}}$.

Fix $N\in\mathbb{N}$, and for $(j_m)_{m\in\mathbb{Z}}\in J$ as above, define the modified position operator
\begin{equation}\label{def:tilde-Q-N}
\tilde{Q}_N= \tilde{Q}_N\big((j_m)_{m\in\mathbb{Z}}\big)
:=\sum_{m=N}^\infty\left( j_{m-N}P_m +j_{-(m-N)}P_{-m}\right).
\end{equation}
We present a proof of the following lemma at the end of this section.

\begin{lemma}\label{lem:Q-bound}
For a given $N\in\mathbb{N}$, and the the operator $\tilde{Q}_N=\tilde{Q}_N\big((j_m)_{m\in\mathbb{Z}}\big)$ defined in \eqref{def:tilde-Q-N}. The difference $(Q-\tilde{Q}_N)$ is $Q$-bounded with relative bound 
\begin{equation}\label{def:b}
b_N\leq 1-\left(1+q\right)^{-(N+1)}<1, \ \text{ where }q=q_{(j_m)}:=\max_\pm\ \limsup_{m\to\pm \infty}\frac{g_m}{|j_m|}.
\end{equation}
\end{lemma}


 Then Theorem \ref{Thm:Q} shows that $D(Q)\subset D(\tilde{Q}_N)$ and  for every $N\in\mathbb{N}$ we have
\begin{equation}\label{eq:v-bound-0}
 v(W)
\leq (1-b_N)^{-1} \sup_{\substack{\psi\in D(Q)\\ \|\psi\|=1}}
    \limsup_{t\to\infty}\frac{1}{t}\bigr\|\tilde{Q}_N\,W^t \psi\bigr\|\leq
    (1+q)^{N+1}\sup_{\substack{\psi\in D(Q)\\ \|\psi\|=1}}\limsup_{t\to\infty}\frac{1}{t}\bigr\|[\tilde{Q}_N, W^t]\psi\bigr\|.
\end{equation} 
Expand the commutator
\begin{equation}\label{eq:com-expand}
[\tilde{Q}_N, W^t]=\sum_{k=0}^t W^k[\tilde{Q}_N, W] W^{t-k-1}
\end{equation}
and since $W$ is unitary, we obtain from \eqref{eq:v-bound-0}, if $[\tilde{Q}_N, W]$  is bounded on $\mathcal{H}$, then
\begin{equation}\label{v-bound-1}
v(W)\leq (1+q)^{N+1}\bigr\|[\tilde{Q}_N, W]\bigr\|=(1+q)^{N+1} \bigr\|[\tilde{Q}_N, \CM]\bigr\|
\end{equation}
where the last equality follows from the fact that $\restr{\tilde{Q}_N}{\CH_m}$ is a constant multiple of the identity and hence it commutes with the block diagonal $\CL$.

A direct tedious calculation using \eqref{def-CM-dec} and \eqref{def:tilde-Q-N} shows that
\begin{align}
[\tilde{Q}_N, \CM]&=
\sum_{m=N}^\infty g_{m-N} \left(a_2(j_{m+1})\ |e_{m+1|1}\rangle\langle e_{m|d_{m}}|-\overline{a_2(j_{m+1}})\ |e_{m|d_m}\rangle\langle e_{m+1|1}|\right) + \notag\\
& + \sum_{m=N+1}^\infty g_{-m+1+N} \left(a_2(j_{-m+1})\ |e_{-m+1|1}\rangle\langle e_{-m|d_{-m}}| - \overline{a_2(j_{-m+1}})\ |e_{-m|d_{-m}}\rangle\langle e_{-m+1|1}|\right)
\end{align}
Then it is direct to see that $[\tilde{Q}_N, \CM]\left([\tilde{Q}_N, \CM]\right)^*=$
\begin{align}
&=
\sum_{m=N}^\infty g_{m-N}^2\ |a_2(j_{m+1})|^2 \left(|e_{m|d_m}\rangle\langle e_{m|d_m}| + |e_{m+1|1}\rangle\langle e_{m+1|1}|\right)\notag\\
& + \sum_{m=N+1}^\infty g_{-m+1+N}^2 |a_2(j_{-m+1})|^2 \left(|e_{-m|d_{-m}}\rangle\langle e_{-m|d_{-m}}| + |e_{-m+1|1}\rangle\langle e_{-m+1|1}|\right)
\end{align}
which is a diagonal matrix, hence (if $[\tilde{Q}_N, \CM]$ is a bounded operator)
\begin{equation}
\bigr\|[\tilde{Q}_N, \CM]\bigr\|= \max\left\{\sup_{m\geq N}g_{m-N}\ |a_2(j_{m+1})| , \sup_{m\geq N+1}g_{-m+1+N}\ |a_2(j_{-m+1})|\right\}
\end{equation}
Note that the left hand side of \eqref{v-bound-1} is an upper bound for $v(W)$ for all $N\in\N$ and for arbitrary set of indices $(j_m)_{m\in\Z}\in J$, thus
\begin{equation}\label{eq:v-bound}
0\leq v(W)\leq \inf_{(j_m)\in J }\inf_{N\in\mathbb{N}}(1+q)^{N+1}\max\left\{\sup_{m\geq N}g_{m-N}\ |a_2(j_{m+1})| ,\sup_{m\geq N+1}g_{-m+1+N}\ |a_2(j_{-m+1})| \right\}.
\end{equation}

To produce a similar bound in terms of $(a_1(n))_{n\in\Z}$ we follow the following argument.

It is straightforward to verify that the commutators 
$[Q,\CL]$, $[Q,\CM]$, and $[Q,T^{\pm1}]$ 
are bounded operators on $\CH$, so by 
 Lemma \ref{lem:com-velocity}, we have
\begin{equation}\label{eq:LM-ML}
v(\CL \CM)
= v(\CM \CL)
= v(T^{-1}\CM \CL T)
= v\bigl((T^{-1}\CM T)(T^{-1}\CL T)\bigr)
=: v(\tilde{\CL}\tilde{\CM}),
\end{equation}
where $T$ denotes the translation by one on $\ell^2(\Z)$. 

The operator $\tilde{\CL}=T^{-1}\CM T$ (shifted $\CM$) is $2\times2$  block-diagonal with respect to
$\bigoplus_{n\in\Z}\Span\{\delta_{2n-2},\delta_{2n-1}\}$,
hence block-diagonal with respect to \eqref{decomposition}, and in particular
$[\tilde Q_N,\tilde{\CL}]=0$. Moreover,
\begin{equation}
\tilde{\CM}
\equiv \bigoplus_{n\in\Z} (\sigma_x C_1(n))
\end{equation}
is $2\times2$ block-diagonal with respect to $
\bigoplus_{n\in\Z}\Span\{\delta_{2n-1},\delta_{2n}\}$,
so that $C_1$ plays the role of the active coin. Then the argument above produces the bound \eqref{eq:v-bound} with $a_2$ replaced by $a_1$.

It remains to prove that $(Q-\tilde{Q}_N)$ is $Q$-bounded with relative bound less than 1. 
\begin{proof}[Proof of Lemma \ref{lem:Q-bound}]

Recall the definitions \eqref{def:Q-Hm} and \eqref{def:tilde-Q-N} of $Q$ and $\tilde{Q}_N$, respectively.

For $m\in\Z$, define $\alpha_\ell^{(m)}$ to be the $\ell$-th diagonal entry of $\restr{(Q-\tilde{Q}_N)}{\CH_m}$,
\begin{equation}
\alpha_\ell^{(m)}:=\left(\restr{(Q-\tilde{Q}_N)}{\CH_m}\right)_{\ell,\ell}, \text{ for }\ell=1,2,\ldots,d_m.
\end{equation}
For $(j_m)\in J$, we have
\begin{equation}
0\leq q_{\pm}:=\limsup_{m\to \pm\infty}\frac{g_m}{j_m} \quad  \iff \quad \limsup_{m\to\infty}\frac{j_{\pm m+1}}{j_{\pm m}}= 1+q_\pm 
\end{equation}

Fix $N\in\N$. Note that for $m\geq N$
\begin{equation}
\frac{|\alpha_\ell^{(m)}|}{(\restr{Q}{\CH_m})_{\ell,\ell}} = \frac{(\restr{Q}{\CH_m})_{\ell,\ell}-j_{m-N}}{(\restr{Q}{\CH_m})_{\ell,\ell}}=1-\frac{j_{m-N}}{(\restr{Q}{\CH_m})_{\ell,\ell}}\leq 1-\frac{j_{m-N}}{j_{m+1}}.
\end{equation}
Take the limit sup as $m\to\infty$ and observe that
\begin{equation}
\displaystyle \liminf_{m\to\infty}\frac{j_{m-N}}{j_{m+1}}=\liminf_{m\to\infty}\prod_{r=0}^{N}\frac{j_{m-r}}{j_{m-r+1}} \geq  \prod_{r=0}^{N}\left(\limsup_{m\to\infty}\frac{j_{m-r+1}}{j_{m-r}}\right)^{-1}=(1+q_+)^{-(N+1)}.
\end{equation}
This shows that 
\begin{equation}\label{eq:alpha-jm-m+}
\limsup_{m\to\infty}\frac{|\alpha_\ell^{(m)}|}{(\restr{Q}{\CH_m})_{\ell,\ell}}\leq 1-(1+q_+)^{-(N+1)}\in [0,1).
\end{equation}
Similarly, for $m\leq -N$
\begin{equation}
\left| \frac{\alpha_\ell^{(m)}}{(\restr{Q}{\CH_m})_{\ell,\ell}}\right| = \frac{|(\restr{Q}{\CH_m})_{\ell,\ell}|-|j_{m+N}|}{|(\restr{Q}{\CH_m})_{\ell,\ell}|}=1-\frac{|j_{m+N}|}{|(\restr{Q}{\CH_m})_{\ell,\ell}|} \leq 1-\frac{j_{m+N}}{j_{m}},
\end{equation}
which shows that
\begin{equation}\label{eq:alpha-jm-m-}
\limsup_{m\to -\infty}\left| \frac{\alpha_\ell^{(m)}}{(\restr{Q}{\CH_m})_{\ell,\ell}}\right|\leq 1-(1+q_-)^{-N}\in [0,1).
\end{equation}
\eqref{eq:alpha-jm-m+} and \eqref{eq:alpha-jm-m-} show that, for all $\ell=1,\ldots,d_m$
\begin{equation}
\limsup_{|m| \to \infty}\left| \frac{\alpha_\ell^{(m)}}{(\restr{Q}{\CH_m})_{\ell,\ell}}\right| \leq 1-\left(1+\max\{q_-,q_+\}\right)^{-(N+1)}=:c_N\in[0,1).
\end{equation}
That is, for any given $\epsilon>0$, there exists an integer $K=K(\epsilon, N)<0$ such that for all $\ell=1,\ldots,d_m$,
\begin{equation}
|\alpha_\ell^{(m)}|<(\epsilon+c_N)|(\restr{Q}{\CH_m})_{\ell,\ell}| \text{ for all }|m|\geq K.
\end{equation}
Then for any $\phi\in D(Q)$ we have
\begin{align}\label{Relative-bound-tilde-Q}
\|(Q-\tilde{Q}_N)\phi\|^2&=\sum_{m\in\Z}\left\|P_m(Q-\tilde{Q}_N)\phi\right\|^2 \notag\\
&= \sum_{|m|<K} \sum_{\ell=1}^{d_m} |\alpha_\ell^{(m)}|^2 |(P_m\phi)_\ell|^2
+ \sum_{|m|\geq K} \sum_{\ell=1}^{d_m} |\alpha_\ell^{(m)}|^2 |(P_m\phi)_\ell|^2\notag\\
&\leq \sup_{\substack{|m|<K\\ \ell=1,\ldots,d_m}}|\alpha_\ell^{(m)}|^2 \ \|\phi\|^2 + (\epsilon+c_N)^2 \sum_{|m|\geq K} \sum_{\ell=1}^{d_m} |(\restr{Q}{\CH_m})_{\ell,\ell}|^2  |(P_m\phi)_\ell|^2
\end{align}
 Set
\begin{equation}
M:=M(\epsilon,N)=\sup_{\substack{m<|K|\\ \ell=1,\ldots,d_m}}|\alpha_\ell^{(m)}|,
\end{equation}
to see that \eqref{Relative-bound-tilde-Q} is bounded as
\begin{equation}
\|(Q-\tilde{Q}_N)\phi\|^2 \leq M^2\|\phi\|^2+(\epsilon+c_N)^2\|Q\phi\|^2.
\end{equation}
Since $c_N\in[0,1)$ and $\epsilon>0$ is arbitrary, then $(Q-\tilde{Q}_N)$ is $Q$-bounded with $Q$-bound less than or equal to $c_N$, where we recall that
\begin{equation}
c_N=1-\left(1+\max\{q_+,q_-\}\right)^{-(N+1)} <1.
\end{equation}
\begin{remark}\label{rem:Q-bound-0}
We note that if $\displaystyle \lim_{|m|\to\infty}g_m/|j_m|=0$, i.e., $q_\pm=0$, then $(Q-\tilde{Q}_N)$ is $Q$-bounded with relative bound zero. 
\end{remark}
\end{proof}

\section{Uniformly bounded and sublinearly growing gaps}\label{sec:zero-velocity-1-2}
In this section, we establish parts (\ref{Thm:case-1}) and (\ref{Thm:case-2}) of Theorem \ref{Thm:determ}. Our argument is based on the a priori velocity bound \eqref{eq:v-bound-general} from Theorem \ref{Thm:general}; for the convenience of the reader, we recall it here.
\begin{equation}
v(W)\leq \min_{k\in\{1,2\}}\inf_{(j_m)\in J}\inf_{N\in\N}\left(1+q\right)^{N+1}
\max\left\{
\sup_{m\geq N} g_{m-N}|a_k(j_{m+1})|,
\sup_{m\geq N} g_{-m+N}|a_k(j_{-m})|
\right\},
\end{equation}
where $J$ is the set of all indices satisfying
\begin{equation}
q=q_{(j_m)}=\max_{\pm}\limsup_{m\to\pm\infty}\frac{g_m}{|j_m|}<\infty.
\end{equation}

Suppose for a fixed $k\in\{1,2\}$, there is a strictly  increasing sequence $(j_m)_{m\in\Z}\subset\Z$ such that $a_k(j_m)\to 0$ as $|m|\to\infty$.

In the case of uniform boundedness of the gaps \eqref{cond-uniform-bound}, $g_m/j_m\to 0$ as $m\to\infty$, hence $(j_m)_{m\in\Z}\in J$, then it follows directly from Theorem \ref{Thm:general}, noting that $q=0$ (see Remark \ref{rem:Q-bound-0}), that
\begin{equation}
0\leq v(W)\leq \big(\sup_{m}g_m\big) \limsup_{|N|\to\infty}|a_k(j_N)| =0,
\end{equation}
which proves part (\ref{Thm:case-1}).

To prove part (\ref{Thm:case-2}), we start by defining $f_k^\pm(N)$ as follows.
\begin{equation}\label{eq:fN}
f_k^+(N):=\sup_{m\geq N}g_{m-N}|a_k(j_{m+1})| \text{ and }f_k^-(N):=\sup_{m\geq N}g_{-m-N}|a_k(j_{-m})|.
\end{equation}
In both cases, we will show that $f_k^+(N)\to 0$ as $N\to\infty$. An identical argument shows that $f_k^-(N)\to 0$ as $N\to\infty$. Hence, $\max\{f_k^+(N), f_k^-(N)\}\to 0$, and Theorem \ref{Thm:general} yields
\begin{equation}
v(W)\leq \min_{k\in\{1,2\}}\inf_{\substack{(j_m)\in J\\ g_m/|j_m|\to 0}}\inf_{N\in\mathbb{N}}\max\{f_k^+(N), f_k^-(N)\}=0.
\end{equation}
We also use the fact that $q=0$, see Remark \ref{rem:Q-bound-0}.

By the assumption $|a_k(j_{m})|=\mathcal{O}(j_m^{-1})$ in (\ref{Thm:case-2}), there exist constants $C>0$ and $m_0\in\N$ such that
\begin{equation}\label{eq:b-bound}
|a_k(j_{m+1})| \leq \frac{C}{j_{m+1}}, \quad \text{for all } m\geq m_0.
\end{equation}
Reindex with $r:=m-N\geq 0$ to see that
\begin{equation}\label{eq:fN-bound1}
f_k^+(N)\leq C \sup_{r\geq 0}\frac{g_r}{j_{N+r+1}} \text{ for all }N\geq m_0.
\end{equation}

Assume that 
\begin{equation}
\quad
\frac{g_m}{|j_m|}\xrightarrow[|m|\to\infty]{}0 \quad \text{ hence } q=0.
\end{equation}
Given $\varepsilon>0$. Since $g_r/j_r\to 0$, there exists $R=R(\epsilon)\in\N$ such that
\begin{equation}\label{eq:choose-R}
\frac{g_r}{j_r}\leq \varepsilon, \text{ for }r\geq R.
\end{equation}
That is, for $N\geq m_0$ we have
\begin{equation}\label{eq:comb}
f_k^+(N)
\leq C
\max\!\left\{
\max_{0\leq r<R}\frac{g_r}{j_{N+r+1}},
\ \sup_{r\geq R}\frac{g_r}{j_{N+r+1}}
\right\}.
\end{equation}
For $r\geq R$, by monotonicity of $(j_m)$ and \eqref{eq:choose-R}, we have
\begin{equation}\label{eq:tail}
\sup_{r\geq R}\frac{g_r}{j_{N+r+1}}
\le
\sup_{r\geq R}\frac{g_r}{j_r}
\leq \varepsilon.
\end{equation}
For the finitely many indices $0\leq r<R$, set
$\displaystyle\max_{0\leq r<R} g_r <\infty$. 
Since $j_m\to\infty$ as $m\to\infty$, there exists $N_1\in\N$ such that for all $N\geq N_1$
\begin{equation}\label{eq:finite}
\max_{0\leq r<R}\frac{g_r}{j_{N+r+1}}
\le
\frac{1}{{j_{N}}}\max_{0\leq r<R} g_r 
\leq \varepsilon
\end{equation}
where we used $j_{N+r+1}>j_N$.

Combining \eqref{eq:tail} and \eqref{eq:finite}, we obtain from \eqref{eq:comb}  that for $N\geq \max\{m_0,N_1\}$,
\begin{equation}
f_k^+(N)\leq C\varepsilon.
\end{equation}
Since $\varepsilon>0$ is arbitrary, we conclude $f_k^+(N)\to 0$ as $N\to\infty$, as desired.

\begin{remark}\label{rem:case-2.5}
We also obtain zero velocity if  we assume that there exists $(j_m)\in J$ such that
\begin{equation}
\limsup_{|m|\to\infty}\frac{g_m}{|j_m|}=c\in(0,\infty) \quad\text{ and }\quad |a(j_m)|=o(j_m^{-1}).
\end{equation}
We give the proof only for the positive direction of the indices, since the argument for the negative direction is identical.

For each $M$, define the set of indices $j^{(M)}_m:=j_{m+M}$. Since $\limsup$ is invariant under shifts,
\begin{equation}
\limsup_{m\to\infty}\frac{g^{(M)}_m}{j^{(M)}_m}
=\limsup_{m\to\infty}\frac{g_{m+M}}{j_{m+M}}=c,
\end{equation}
so each tail sequence $(j^{(M)}_m)$ still belongs to $J$.

Apply the general bound \eqref{eq:v-bound-general} with $N=1$ to the sequence $(j^{(M)}_m)$ gives
\begin{equation}
v(W)\le (1+q)^2\sup_{m\ge 0} g_{m+M-1}|a(j_{m+M+1})|.
\end{equation}
Taking the infimum over $M$ yields
\begin{equation}
v(W)\le (1+q)^2\inf_{M\ge 0}\sup_{k\geq M} g_{k-1}|a(j_{k+1})|
=(1+q)\limsup_{k\to\infty} g_k|a(j_{k+1})|
=0.
\end{equation}
The last step follows from observing that $g_k=j_{k+1}-j_k\leq j_{k+1}<j_{k+2}$, thus
\begin{equation}
g_{k-1} |a(j_{k+1})|\leq j_{k+1} |a(j_{k+1})|\to 0.
\end{equation}
Therefore $v(W)=0$.
\end{remark}

\section{Gap-weighted decay}\label{sec:pf:case3}

In this section we prove part (\ref{Thm:case-3}) of Theorem \ref{Thm:determ}.

Consider the Hilbert space decomposition \eqref{dec}, and recall that for $m\in\Z$, we define $P_m$ to be the orthogonal projection onto $\mathcal{H}_m$. For general subsets $\Lambda\subset\Z$, let $P_\Lambda$ be the orthogonal projection onto $\bigoplus_{m\in \Lambda}\CH_m$, i.e., 
\begin{equation}
P_\Lambda:=\sum_{m\in \Lambda}P_m.
\end{equation}

For any fixed $N\in\N$, we have
\begin{equation}\label{eq:case3-first}
v(W)\leq \sup_{\substack{\psi\in D(Q)\\ \|\psi\|=1}}\limsup_{t\to\infty}\frac{1}{t}\left(\left\|P_{(-N,N)}Q W^t \psi\right\|+\left\|P_{\Z\setminus(-N,N)}Q W^t \psi\right\|\right)
\end{equation}
Here $(-N,N)$ is short for $(-N,N)\cap\Z$.
Observe that, since $W$ is unitary and $\|\psi\|=1$,
\begin{equation}
\frac{1}{t}\left\|P_{(-N,N)}Q W^t \psi\right\|\leq \frac{1}{t}\left\|P_{(-N,N)}Q\right\| = \frac{\max\{j_N, |j_{-N}|\}}{t} \ \xrightarrow[t\to\infty]{} 0,
\end{equation}
and that $Q$ and $|Q|$ differ only by a unitary. Hence, back to \eqref{eq:case3-first}, we obtain
\begin{equation}\label{iii-2}
v(W)\leq \sup_{\substack{\psi\in D(Q)\\ \|\psi\|=1}}\limsup_{t\to\infty}\frac{1}{t}\left\|P_{\Z\setminus(-N,N)}|Q| W^t \psi\right\|
\end{equation}
Next consider the (positive) operator
\begin{equation}
\hat{Q}_N=\sum_{m=N}^\infty\left( j_{m+1}P_m +|j_{-m}|P_{-m}\right)
\end{equation}
that takes the largest value of $\restr{\left(|Q|\right)}{\CH_k}$ in each block with $m\in\Z\setminus(-N,N)$. Then  
it is clear that $P_{\Z\setminus(-N,N)}|Q|\leq \hat{Q}_N$, hence
\begin{equation}
\bigr\|P_{\Z\setminus(-N,N)}|Q| W^t\psi\bigr\|^2=\langle W^t\psi, P_{\Z\setminus(-N,N)}|Q|^2 W^t\psi\rangle\leq \langle W^t\psi, \hat{Q}^2_N W^t\psi\rangle=\bigr\|\hat{Q}_N W^t\psi\bigr\|^2
\end{equation}
Thus, from \eqref{iii-2}, we conclude
\begin{equation}\label{case3-1}
v(W)\leq \limsup_{t\to\infty}\frac{1}{t}\| [\hat{Q}_N, W^t]\|\leq \| [\hat{Q}_N, W]\|= \| [\hat{Q}_N, \CM]\| \text{ for all }N\in \mathbb{N}.
\end{equation}
Here we expanded the commutator as in \eqref{eq:com-expand}, then used $W\equiv \CL\CM$ and $\CL$ is block diagonal with respect to the decomposition \eqref{dec}, hence $[\hat{Q}_N, W]=\CL[\hat{Q}_N, \CM]$.
A direct calculation, using the relabeling \eqref{def:e-m-j}, shows that
\begin{align}
[\hat{Q}_N, \CM]=&
j_{N+1} \left(a_2(j_N) |e_{N|1}\rangle\langle e_{N-1|d_{N-1}}|- \overline{a_2(j_N)} |e_{N-1|d_{N-1}}\rangle\langle e_{N|1}|\right)+ \notag\\[0.1cm]
&+|j_{-N}|\left(\overline{a_2(j_{-N+1}}) |e_{-N|d_{-N}}\rangle\langle e_{-N+1|1}|- a_2(j_{-N+1}) |e_{-N+1|1}\rangle\langle e_{-N|d_{-N}}|\right)+ \notag\\
&+\sum_{m=N}^\infty g_{m+1}\left(a_2(j_{m+1})|e_{m+1|1}\rangle\langle e_{m|d_k}| \ - \ \overline{a_2(j_{m+1}})|e_{m|d_k}\rangle\langle e_{m+1|1}| \right)+\notag\\
& +  \sum_{m=N+1}^\infty g_{-m+1}\left(\overline{a_2(j_{-m+1}})|e_{-m|d_{-m}}\rangle\langle e_{-m+1|1}| - a_2(j_{-m+1})|e_{-m+1|1}\rangle\langle e_{-m|d_{-m}}|\right)
\end{align}
where we used the fact that $g_{-m+1}=j_{-m+1}-j_{-m}=|j_{-m}|-|j_{-m+1}|$ for $m>0$.
Then $[\hat{Q}_N, \CM]([\hat{Q}_N, \CM])^*$ is diagonal and given as
\begin{align}
=& j_{N+1}^2\ |a_2(j_N)|^2 \left( |e_{N|1}\rangle\langle e_{N|1}|+ |e_{N-1|d_{N-1}}\rangle\langle e_{N-1|d_{N-1}}|\right)+ \notag\\[0.1cm]
& +j_{-N}^2\ |a_2(j_{-N+1}) |^2\left(|e_{-N|d_{-N}}\rangle\langle |e_{-N|d_{-N}}| +  |e_{-N+1|1}\rangle\langle e_{-N+1|1}|\right)+ \notag\\
&+\sum_{m=N}^\infty g_{m+1}^2\ |a_2(j_{m+1})|^2\left(|e_{m+1|1}\rangle\langle |e_{m+1|1}| + |e_{m|d_k}\rangle\langle e_{m|d_k}| \right)+\notag\\
& +  \sum_{m=N+1}^\infty g_{-m+1}^2\ |a_2(j_{-m+1})|^2 \left(|e_{-m|d_{-m}}\rangle\langle e_{-m|d_{-m}}| + |e_{-m+1|1}\rangle\langle e_{-m+1|1}|\right).
\end{align}
Hence $\big\|[\hat{Q}_N, \CM]\big\|$ is 
\begin{equation}
= \max \left\{ j_{N+1}|a_2(j_N)|, |j_{-N} a_2(j_{-N+1})|,\sup_{m\geq N}g_{m+1} |a_2(j_{m+1})|, \sup_{m\geq N+1}g_{-m+1} |a_2(j_{-m+1})|\right\}
\end{equation}
Observe that $g_{m+1}\leq j_{m+2}$ and $g_{-m+1}\leq |j_{-m}|$. So we get
\begin{equation}
\big\|[\hat{Q}_N, \CM]\big\| = \max\left\{ \sup_{m\geq N}j_{m+1}|a_2(j_{m})|, \sup_{m\geq N}|j_{-m}||a_2(j_{-m+1})|\right\}
\end{equation}
Following from \eqref{case3-1}, we have
\begin{equation}\label{eq:weighted-gap-a2}
v(W)\leq \max\left\{\limsup_{m\to\infty}j_{m+1}|a_2(j_{m})|, \limsup_{m\to\infty}|j_{-(m+1)}||a_2(j_{-m})| \right\}
\end{equation}
Finally, for $m>0$, we have
\begin{equation}
j_{m+1}\lvert a_2(j_m)\rvert = (g_m + j_m)\lvert a_2(j_m)\rvert.
\end{equation}
Therefore, $v(W)=0$ if and only if
\begin{equation}
g_m\lvert a_2(j_m)\rvert \to 0
\qquad \text{and} \qquad
j_m\lvert a_2(j_m)\rvert \to 0.
\end{equation}
The same argument applies for negative values of $m$.  To prove \eqref{eq:weighted-gap-a2} in terms of $a_1$, we follow the same argument around \eqref{eq:LM-ML}.

This completes the proof of \eqref{Thm:case-3}.

\section{Proof of the random setting}\label{sec:pf-random}
In this section, we prove Corollary \ref{coro:random} using case (\ref{Thm:case-2}) of Theorem \ref{Thm:determ}. The key step is the following lemma, proved by means of a Borel--Cantelli argument. As explained in Remark \ref{rem:random-sharp-alpha-1}, the conclusion is essentially sharp in view of cases (\ref{Thm:case-1})--(\ref{Thm:case-3}) of Theorem \ref{Thm:determ}.
\begin{lemma}\label{lem:good-subsequence}
Let $(a(n))_{n\in\mathbb{Z}}$ be i.i.d.\ random variables from a probability space $(\Omega,\mathcal{B},\mathbb{P})$ 
satisfying
\begin{equation}\label{eq:random-assumption}
\mathbb{P}\left(\left\{\omega:\ |a(0;\omega)|\leq x\right\}\right) \geq c x^\alpha \quad \text{for all sufficiently small } x>0
\end{equation}
for some $c>0$ and $\alpha\in(0,1)$.
Then, almost surely, there exists a bi-infinite strictly increasing sequence
$(j_m)_{m\in\mathbb{Z}} \subset \mathbb{Z}$ such that
\begin{equation}\label{pf-random}
 \frac{g_m}{|j_m|} \xrightarrow[|m|\to\infty]{} 0 \text{\ and \ }|a(j_m)| \leq \frac{1}{|j_m|} \text{ for all }m\in\Z. 
\end{equation}
\end{lemma}

\begin{proof}
We present the proof for the positive indices. The argument for negative indices is identical, applied to the i.i.d.\ sequence
$(a(-n))_{n\ge1}$.

For $k\in\Z_+$, set 
\begin{equation}\label{def:E}
E_k := \{\omega:\ |a(k;\omega)| \leq 1/k \}.
\end{equation}
By \eqref{eq:random-assumption} and identical distribution, we have 
$\mathbb{P}\left(E_k\right) \geq c k^{-\alpha}$  for sufficiently large $n$.
Since $\alpha<1$, then
\begin{equation}
\sum_{k=1}^\infty \mathbb{P}(E_k)= \infty.
\end{equation}
The events $E_k$ are independent, hence by the second part of the
Borel-Cantelli lemma, 
\begin{equation}
\mathbb{P}(\Omega_E)=1, \text{ where }\Omega_E:=\{E_k \text{ occurs infinitely often}\}.
\end{equation}
Thus, on $\Omega_E$, there exists an increasing sequence $(j_m)_{m\ge1}$
such that $E_{j_m}$ holds for all $m$, i.e.,
\begin{equation}\label{pf:random-decay}
|a(j_m)| \leq \frac{1}{j_m} \qquad \text{for all }m\ge1.
\end{equation}

Next we show that the gaps $g_m=j_{m+1}-j_m$ satisfy the first conclusion of \eqref{pf-random} on the positive side.
For $n\geq 2$, consider the events 
consider the events
\begin{equation}\label{def:B}
B_n :=E_n\cap\bigcap_{k=n+1}^{n+h(n)}E_k^c, \text{\ where \ }h(n):=\bigl\lceil n^{\alpha}(\ln n)^2\bigl\rceil .
\end{equation}
i.e., $B_n$ is the event that $E_n$ occurs and all $E_{n+1},E_{n+2},\ldots, E_{n+h(n)}$ do not occur.

If $B_n$ occurs, then $n$ is a ``good index'' and the next good index
(when it exists) is at least $n+h(n)$, so the gap following $n$ is at
least $h(n)$.

We will show later that $\sum_{n=1}^\infty \mathbb{P}(B_n)$ converges.
Then by the first part of Borel-Cantelli, almost surely only finitely many
events $B_n$ occur, i.e.,
\begin{equation}
\mathbb{P}(\Omega_B):=1, \text{ where }\Omega_B:=\{B_n \text{ occurs only finitely often}\}.
\end{equation}
Now, fix $\omega\in \Omega_E\cap\Omega_B$ and let $(j_m)_{m\ge1}$ be the increasing sequence of all positive indices for which $E_{k}$
occurs. Then there exists $N\in\mathbb{N}$ such that for every $n\geq N$, if $E_n$ occurs, then $E_k$ also occurs with $k\in\{n+1,\ldots,n+h(n)\}$.
This implies that for all sufficiently large $m$ (so that $j_m\geq N$),
\begin{equation}
g_m \leq h(j_m).
\end{equation}
Dividing by $j_m$ and using $\alpha<1$ gives
\begin{equation}
\frac{g_m}{j_m}\leq \frac{j_m^{\alpha} (\ln j_m)^2+1}{j_m}
= j_m^{\alpha-1} (\ln j_m)^2 +\frac{1}{j_m}\xrightarrow[m\to\infty]{} 0.
\end{equation}
It remain to prove that $\sum_{n=1}^\infty \mathbb{P}(B_n)$ converges.
Because the events $E_k$ are independent, we have from \eqref{def:B}
\begin{equation}\label{pf:random-sum-B-1}
\mathbb{P}(B_n) =  \mathbb{P}(E_n) \prod_{k=n+1}^{n+h(n)} (1- \mathbb{P}(E_k)) \leq  
\exp\left(-\sum_{k=n+1}^{n+h(n)}  \mathbb{P}(E_k)\right)
\end{equation}
where we used $1-x\leq e^{-x}$ and $\mathbb{P}(E_n)\leq 1$ in the last step.

For  $k\in[n+1,n+h(n)]$ we have, by \eqref{eq:random-assumption}
\begin{equation}
 \mathbb{P}(E_k)\geq c k^{-\alpha}\geq c(n+h(n))^{-\alpha}.
\end{equation}
Note that $h(n)/n\to 0$ as $n\to\infty$. Hence for all sufficiently large $n$, $h(n)\leq n$, and  we obtain
\begin{equation}
 \mathbb{P}(E_k)\geq c k^{-\alpha}\geq c(2n)^{-\alpha}\geq \frac{1}{2}c n^{-\alpha}.
\end{equation}
Then we  obtain (for sufficiently large $n$)
\begin{equation}
\sum_{k=n+1}^{n+h(n)}  \mathbb{P}(E_k) \geq \frac{1}{2}c n^{-\alpha} \sum_{k=n+1}^{n+h(n)} 1= \frac{1}{2}c n^{-\alpha}h(n)= \frac{1}{2}c\ (\ln n)^2.
\end{equation}
Use this bound in \eqref{pf:random-sum-B-1} to obtain
\begin{equation}
\mathbb{P}(B_n)\leq \exp\left(-\frac{1}{2}c\ (\ln n)^2\right)=n^{-\frac12 c \ln n}\leq n^{-2}, \text{ for all sufficiently large }n,
\end{equation}
and the series $\sum_{n=1}^\infty \mathbb{P}(B_n)$ converges.

\end{proof}

\begin{remark}\label{rem:random-sharp-alpha-1}
The restriction $\alpha<1$ is essential. In general, the conclusion fails for
$\alpha=1$. For example, let $(a(n))_{n\in\mathbb Z}$ be i.i.d.\ uniform on
$[-1,1]$. Then
\begin{equation}\label{random:eq:rem1}
\mathbb P(|a(0)|\leq x)=x \qquad (0<x\leq 1),
\end{equation}
so \eqref{eq:random-assumption} holds with $\alpha=1$ and $c=1$. 
\begin{enumerate}[(A)]\itemsep0.5cm
\item
Writing
$E_k:=\{|a(k)|\leq 1/k\}$, we have by \eqref{random:eq:rem1}, $\mathbb P(E_k)=1/k$. Define
\begin{equation}
A_n:=\bigcap_{j=2^n}^{2^{n+1}-1} E_k^c.
\end{equation}
The events $A_n$ are independent and
\begin{equation}
\mathbb P(A_n)
=
\prod_{j=2^n}^{2^{n+1}-1}\left(1-\frac1k\right)
=\prod_{j=2^n}^{2^{n+1}-1}\frac{k-1}{k}=
\frac{2^n-1}{2^{n+1}-1}.
\end{equation}
Hence $\sum_n \mathbb P(A_n)=\infty$, so by Borel-Cantelli, almost surely
infinitely many blocks $[2^n,2^{n+1})$ contain no good index. It follows
that no increasing sequence of good indices can satisfy $g_m/j_m\to 0$ (because the indices are growing at least as $\sim 2^m$).

\item
More generally, the same example shows that the endpoint $\alpha=1$ cannot, in
general, be rescued by replacing the condition \eqref{pf-random} with
\begin{equation}
\lim_{m\to\infty}\frac{g_m}{j_m}\in(0,\infty), \text{ and }|a(j_m)|\leq f(j_m),
\end{equation}
where $f:\mathbb N\to(0,1]$ satisfies
\begin{equation}
f(n)=o(n^{-1}) \qquad \text{ as }n\to\infty.
\end{equation}

Indeed, let $(a(n))_{n\in\mathbb Z}$ be i.i.d.\ uniform on $[-1,1]$, and define
\begin{equation}
E_k:=\{|a(k)|\leq f(k)\}, \qquad k\ge1.
\end{equation}
By \eqref{random:eq:rem1}, we have $\mathbb P(E_k)=f(k)$.

Hence, if $\sum_{k=1}^\infty f(k)<\infty$,
then also $\sum_k \mathbb P(E_k)<\infty$, and so by the first
Borel--Cantelli lemma, almost surely only finitely many $E_k$ occur. In
particular, there is then no (infinite) subsequence $(j_m)\subset\Z$ satisfying $|a(j_m)|\leq f(j_m)$.

It therefore remains only to consider the case
$\sum_{k=1}^\infty f(k)=\sum_{k=1}^\infty \mathbb P(E_k)=\infty$.
Hence, by the second Borel-Cantelli lemma, almost surely $E_k$ occurs
infinitely often.

Now fix an integer $\ell\ge2$, and define the intervals
\begin{equation}
I_n^{(\ell)}:=[\ell^n,\ell^{n+1})\cap\mathbb N,
\qquad
A_n^{(\ell)}:=\bigcap_{k\in I_n^{(\ell)}} E_k^c.
\end{equation}
Since the blocks $I_n^{(\ell)}$ are
disjoint, then $\{A_n^{(\ell)}\}_n$ are independent events,  Moreover, 
\begin{align}
\mathbb{P}\bigl((A_n^{(\ell)})^c\bigr)
\leq& \sum_{k\in I_n^{(\ell)}} \mathbb P(E_k) = \sum_{k\in I_k^{(\ell)}} f(k) \notag\\
\leq& \bigl(\sup_{r\in I_n^{(\ell)}} r f(r)\bigr) \sum_{k\in I_n^{(\ell)}}\frac{1}{k} \notag\\
\leq&  \bigl(\sup_{r\geq \ell^n} r f(r)\bigr) \int_{\ell^n-1}^{\ell^{n+1}}\frac{dx}{x}  \notag\\
=& \bigl(\sup_{r\geq \ell^n} r f(r)\bigr)\left(\ln \ell+\ln\frac{\ell^n}{\ell^n-1}\right) \xrightarrow[n\to\infty]{} 0.
\end{align}
Here we used the fact that $f(n)=o(n^{-1})$, hence for each fixed $\ell\ge 2$ we have
\begin{equation}
\sup_{r\in I_n^{(\ell)}} r f(r)\leq \sup_{r\geq \ell^n} r f(r)\xrightarrow[n\to\infty]{}0.
\end{equation}
Thus $\mathbb P(A_n^{(\ell)})\to1$ as $n\to\infty$, and in  particular,
\begin{equation}
\sum_{n=1}^\infty \mathbb P(A_n^{(\ell)})=\infty.
\end{equation}
Since the events $A_n^{(\ell)}$ are independent, the second
Borel-Cantelli lemma implies that $A_n^{(\ell)}$ occurs infinitely often almost
surely.

For each integer $\ell\ge2$, let
\begin{equation}
\Omega_\ell:=\{A_n^{(\ell)} \text{ occurs infinitely often}\}.
\end{equation}
Then $\mathbb P(\Omega_\ell)=1$ for every $\ell\ge2$, and hence
\begin{equation}
\mathbb P(\Omega_*)=1 \text{ where }\Omega_*:=\bigcap_{\ell=2}^\infty \Omega_\ell.
\end{equation}

Now fix $\omega\in\Omega_*$ and suppose there exists an increasing sequence of
good indices $(j_m)$ such that
\begin{equation}\label{pf:random-3}
|a(j_m)|\leq f(j_m)
\qquad\text{and}\qquad
\frac{g_m}{j_m}\xrightarrow[m\to\infty]{}L\in(0,\infty).
\end{equation}
Choose an integer $M>1+L$. Since $A_n^{(M)}$ occurs infinitely often, there
are infinitely many $n$ for which the block $I_n^{(M)}$ contains no good index.
For all sufficiently large such $n$, there exists $m$ with
\begin{equation}
j_m<M^n<M^{n+1}\leq j_{m+1}.
\end{equation}
Hence
\begin{equation}
\frac{g_m}{j_m}
=
\frac{j_{m+1}-j_m}{j_m}
>
\frac{M^{n+1}-M^n}{M^n}
=
M-1
>
L
\end{equation}
for infinitely many $m$, contradicting our assumption in \eqref{pf:random-3} that 
$g_m/j_m\to L$.

\end{enumerate}
\end{remark}

\bibliographystyle{abbrvArXiv}
\bibliography{zero-velocity.bib}
\end{document}